\documentclass[aps, superscriptaddress, twocolumn]{revtex4}

\usepackage{scalefnt}
\usepackage{amsmath,amsthm,amsfonts,amssymb} 
\usepackage{fontenc}
\usepackage[all]{xy}
\usepackage{graphicx,subfigure,multirow}
\usepackage{tikz}
\usetikzlibrary{shapes, arrows, shadows}
\usepackage[latin1]{inputenc}
\usepackage[scaled]{helvet}
\usepackage[toc, page]{appendix}
\usepackage{graphicx}				
\usepackage{amssymb}
\usepackage{changes}
\usepackage{lipsum}

\newtheorem{theorem}{Theorem}
\newtheorem{proposition}{Proposition}
\newtheorem{lemma}{Lemma}

\newtheorem{example}{Example}

\input{Qcircuit.tex}
\begin{document}
\title{Noise in One-Dimensional Measurement-Based Quantum Computing}
\author{Na\"iri Usher}
\email{ucapnus@ucl.ac.uk}
\affiliation{Department of Physics and Astronomy, University College London, Gower Street, London WC1E 6BT, United Kingdom.}
\author{ Dan E. Browne}
\affiliation{Department of Physics and Astronomy, University College London, Gower Street, London WC1E 6BT, United Kingdom.}
\begin{abstract}
Measurement-Based Quantum Computing (MBQC) is an alternative to the quantum circuit model, whereby the computation proceeds via measurements on an entangled resource state. Noise processes are a major experimental challenge to the construction of a quantum computer. Here, we investigate how noise processes affecting physical states affect the performed computation by considering MBQC on a one-dimensional cluster state. This allows us to break down the computation in a sequence of building blocks and map physical errors to logical errors. Next, we extend the Matrix Product State construction to mixed states (which is known as Matrix Product Operators) and once again map the effect of physical noise to logical noise acting within the correlation space.  This approach allows us to  consider more general errors than the conventional Pauli errors, and could be used in order to simulate noisy quantum computation. 
\end{abstract}
%\pacs{ 03.67.-a, }
% 03.65.-w	Quantum mechanics
% 03.67.-a Quantum information
%03.65.Aa	Quantum systems with finite Hilbert space
%03.65.Ta	Foundations of quantum mechanics; measurement theory (for optical tests of quantum theory, see 42.50.Xa)
% 03.65.Ud	Entanglement and quantum nonlocality 
%03.67.-a	Quantum information
%03.67.Ac	Quantum algorithms, protocols, and simulations
\maketitle
\section{Introduction}
%Quantum Computing and Noise
Quantum computing is a novel and powerful paradigm of computation, whereby quantum algorithms offer the possibility of a speed-up over their classical counterpart \cite{shor1999polynomial}. However, a major challenge to the construction of a universal and scalable quantum computer is the sensitivity to noise of quantum states and operations. Indeed, our current models of computation and algorithms are developed for noiseless systems, that is, for pure states undergoing unitary evolution. In practice, noise processes affect the computation from preparation to measurement, thus corrupting the computation and rendering the output useless. Hence, it is important to understand how noise acting on physical states and operations affects the computation at hand.

%Problem Statement 
Measurement Based Quantum Computing (MBQC) \cite{nielsen2006cluster, raussendorf2001one, briegel2009measurement, jozsa2006introduction, raussendorf2003measurement} is an alternative model of computation to the circuit model, whereby the computation is implemented by performing single qubit adaptive measurements on an entangled resource state. The resource state considered is the cluster state, which consists of qubits initialised in the $|+\rangle$ state arranged on a lattice structure, where neighbouring qubits are entangled via a controlled-$Z$ gate. By performing single qubit measurements in either the computational basis or onto the equatorial plane of the Bloch sphere, a unitary operator is applied to the input state, up to a random by-product operators dependent on the measurement outcome. 

%One-dim 
One-dimensional MBQC refers to MBQC performed on a cluster state of dimension one, i.e. qubits placed on a line. It has been shown that a computation performed on a one-dimensional resource state is classically efficiently simulatable  \cite{van2007classical, vidal2003efficient, jozsa2006simulation, markov2008simulating}, and thus two or three dimensional resource states are typically considered. Nonetheless, the study of one-dimensional MBQC has allowed for interesting results, such as for different families of resource states to be studied, and in addition, it has notably been shown that one-dimensional MBQC can be represented with Matrix Product States (MPS) \cite{perez2006matrix, gross2007measurement, gross2007novel, gross2010quantum}.

MPS belong to the family of tensor network states, which include Projected Entangled Pair States (PEPS) and Multi-Scale Entanglement Renormalization Ansatz (MERA) and which refer to classes of quantum states which can be described using a tensor network. There has long been a close relationship between tensor networks and MBQC. Indeed, the PEPS formalism itself \cite{verstraete2004valence} was first introduced as a tool for the study of MBQC, and its relationship with gate teleportation.

It has been shown that tensor network states offer good approximations to ground states of certain common physical Hamiltonians, as illustrated by the success of Density Matrix Renormalization Group (DMRG) numerical methods, which is a variational algorithm over MPS \cite{verstraete2006matrix}. In addition, these can be used to describe the state of a system and its evolution over time, and can also be diagrammatically represented. The dimension of the tensor network depends on the complexity of the state, as for instance determined by the locality of its interactions or its entanglement. 

MPS is a tensor network framework for one-dimensional quantum states, which is efficient in the case when the amount of bipartite entanglement is bounded. Here, the complex coefficients are expressed as a product of a polynomial number of matrices. The size of these matrices is determined by a parameter called the bond size or Schmidt rank, a quantity connected to the amount of bipartite entanglement present in the system. Vidal \cite{vidal2003efficient} showed that, when the size of these matrices is bounded, MPS states can be used to classically efficiently simulate slightly entangled quantum computation. 

Beyond pure states and unitary or isometric evolution, tensor networks have also provided powerful tools for the study of mixed states and noisy, non-unitary evolution. For example, the Matrix Product Operators (MPO) framework, generalises MPS from states to (density) operators \cite{zwolak2004mixed, verstraete2004matrix, pirvu2010matrix}.

%"This result has clear implications" i.e. why important 
In the following, we consider one-dimensional MBQC and introduce two frameworks for one-dimensional mixed state computation. This in turn allows us to study how physical noise acting on quantum states is mapped to logical errors acting on the computation. Here,  noise is modelled as the application of a local, single qubit, noise channel to the state, as these are simple yet relevant for experimental quantum computing. Moreover, whereas error correction schemes focus on the occurrence of Pauli errors, our frameworks allow for the analysis of more general noise models, and can thus be used to simulate noisy quantum channels. 

% First result 
Our first result is a theorem mapping local noise in a one-dimensional MBQC to errors on the output of the computation. This shows us that even in the case of what we call a noisy resource state or a noisy measurement, the logical errors acting on the the output of the computation can be determined. This is achieved by breaking down the computation into what we call fundamental blocks, which can subsequently be composed in order to form the complete computation. This simplifies the task at hand, and allows us to study the action of local noise channels on the directly on the fundamental block itself. 

%Second result 
Our second result is twofold: first, we propose an expression for Matrix Product Operators (MPOs), closely related to what was proposed in \cite{zwolak2004mixed, verstraete2004matrix, pirvu2010matrix}. Next, we use this framework in order to study one-dimensional noisy MBQC. More precisely, we will express the cluster state in this framework and study how local, single qubit noise channels acting on the physical qubits transform the associated logical superoperators acting on the correlation space. 

%Previous work 
The standard approach towards achieving fault tolerance is to devise error corrections codes whereby errors are first detected and then corrected, thus allowing us to recover the framework of pure states undergoing unitary operations. On the other hand, in the one clean qubit model of computation \cite{shor2007estimating}, an $n$-qubit maximally mixed state is used as a resource in order to perform an interesting task, which could not classically be achieved. Thus, we can ask whether other classes of noisy resource states and computations could be exploited in order to perform a non-trivial computation. In the following, we adopt the approach of understanding the effect of noise on the computation, and argue that such a model could be used in order to simulate noisy quantum computation.

%Paper structure 
This paper is structured as follows. In section \ref{sec:motivation}, we consider quantum teleportation in the presence of a noisy resource state. This provides us, in section \ref{sec:noisyMBQC}, with a motivation to more generally study one-dimensional MBQC in the context of mixed resource states. By breaking down the computation in a single building block, we map the effect of local Pauli operators acting on the state to noise on the output. In section \ref{sec: MPO}, we derive an expression for MPOs, and subsequently, in section \ref{sec:MPOMBQC} study one-dimensional MBQC in the MPO framework, thus finding how local noise acting on the physical qubits transform the logical superoperators acting within the correlations space. Finally, in section \ref{sec:discussion} we conclude with a discussion of our results and further research directions to be investigated. 

\section{Motivation} \label{sec:motivation}
The entanglement present in a multi-qubit state is a resource which can be exploited for quantum computation and communication tasks. For example, quantum teleportation allows for information to be transmitted by utilising the entanglement present in the Bell state $|B_{00}\rangle = \frac{1}{\sqrt{2}} (|00\rangle + |11\rangle)$ \cite{bennett1993teleporting}. Here, Alice and Bob are each in possession of one of the two qubits, and in addition Alice has an arbitrary single qubit mixed state  $\rho= \sum_{u, v} \alpha_{uv}|u \rangle \langle v|$ which she wishes to teleport to Bob. By performing a Bell measurement on her pair of qubits, Alice can transmit the unknown state $\rho$ to Bob, up to some additional Pauli operators dependent on the measurement outcome. If she thus sends the two bits of information corresponding to her measurement outcomes to Bob, this will allow him to undo the Pauli operators, thus leaving him in possession of the original unknown state $\rho$. 
\begin{figure}[h!]
	\centering
	\begin{tikzpicture}
	\definecolor{qu}{RGB}{77, 77, 255}
	\definecolor{cz}{RGB}{254, 68, 0}
	\draw[cz] (2,0)--(4,0);
	\node[fill, scale = 0.7] at (0,0) [circle,draw] {};
	\node[fill, scale = 0.7] at (2,0) [circle,draw] {};
	\node[fill, scale = 0.7] at (4,0) [circle,draw] {};
	{
		\scalefont{1}
		\node[text=black] at (1, -1) {Alice};
		\node[text=black] at (4.5, -1) {Bob};
	}
	{
		\scalefont{0.6}
		\node[text=black] at (0, -0.3) {$1$};
		\node[text=black] at (2, -0.3) {$2$};
		\node[text=black] at (4, -0.3) {$3$};
	}
	\draw (1,0) node[dotted, minimum height=1cm,minimum width=3cm,draw] {};
	\end{tikzpicture}
	\caption[Teleportation.]{Alice is in possession of the two leftmost qubits (labelled $1$ and $2$), on which she performs a Bell measurement (illustrated with the dotted box), whilst Bob has the qubit on the right (labelled $3$). The qubits $2$ and $3$ are the entangled resource state.}
\end{figure}
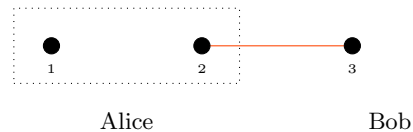

Crucially, this protocol relies on the ability to prepare the pure maximally entangled Bell state $|B_{00}\rangle $. In the following, we ask what would be the effect of instead having access to a noisy mixed state as a resource. This is modelled by the application of a single qubit local Pauli channel $\varepsilon$ to the originally entangled state $|B_{00}\rangle $. 

More precisely, we consider a noisy process whereby the Pauli operators $X^k Z^l$ are applied with probability $p_{kl}$ to the second qubit of the entangled qubit pair. This results in the original resource state evolving to one of the four orthonormal Bell states $|B_{ij}\rangle =\frac{1}{\sqrt{2}}(\mathbb{I} \otimes X^i Z^j)(|00\rangle + {11 \rangle})$, and the resource state is thus now given by $ \Lambda  =  \sum_{i,j} \frac{1}{2} |i\rangle \langle j| \otimes \varepsilon (|i\rangle \langle j| )$. This channel is assumed to be a Completely Positive Trace Preserving (CPTP) map, thus having Kraus decomposition $\varepsilon (\rho)= \sum_m K_m \rho K_m^\dag$, where the Kraus operators must satisfy $\sum_m K_m^\dag K_m=\mathbb{I}$. Here, the Kraus operators are given by $K_{kl}=\sqrt{p_{kl}} X^kZ^l$, and thus the resource state can equivalently be written as: 
\begin{equation}
\Lambda  = \sum_{k,l} \sum_{i,j} \frac{1}{2} p_{kl} (\mathbb{I}\otimes X^k Z^l) |ii \rangle \langle jj| (\mathbb{I}\otimes Z^l X^k). 
\end{equation}
Henceforth, we shall refer to such a state as a diagonal resource state. Indeed, the set of four Bell states form an orthonormal basis, and thus any arbitrary two-qubit state can be expressed as $\rho = \sum \rho_{ijkl}|B_{ij} \rangle \langle B_{kl}|$, with the special case $\rho = \sum_{ij} \rho_{ij} |B_{ij} \rangle \langle B_{ij}|$, being viewed as diagonal. 

In the case of having access to a noisy diagonal resource state, we find that if $\rho$ is the unknown state we wish to teleport, then this will instead result in the state  $\varepsilon (\rho)$ being teleported, where $\varepsilon$ refers to the original noise channel applied to the resource. Thus, this situation is equivalent to a noise channel acting directly on the arbitrary single qubit input state. More formally:  
\begin{proposition} \label{prop: diag}
Given a resource state diagonal in the Bell basis $\varepsilon(\rho)$, the application of the teleportation protocol results in the state $X^s Z^t  \varepsilon (\rho )  Z^t X^s$ being teleported. 
\end{proposition}

In order to further understand the trade-off taking place between the amount of noise affecting the resource state and the entanglement present in the system to be exploited, we apply Proposition \ref{prop: diag} to a couple of examples. First, we consider the worst case scenario of a completely noisy resource state, and find that, as expected, no information will be transmitted. 
\begin{example}
The application of the completely depolarising channels, with corresponding Kraus operators $K_{ij}=\frac{1}{4}\sigma_{ij}$, for $i, j = 0, 1$ results in the resource state being transformed to the maximally mixed state. From proposition $1$, the completely depolarising channel is applied to the transmitted  state $\rho$. Thus, it will now be transformed to the maximally mixed state. This gives us absolutely no information as to what the unknown state $\rho$ was, and our interpretation is that nothing has been teleported. 
\end{example}
Next, we consider the case when a dephasing channel with respect to the computational basis is applied, and find that in this case information has been teleported. 
\begin{example}
When the dephasing channel with respect to the computational basis is applied, the resource state gets mapped to an equal mixture of the two Bell states $|B_{00}\rangle = \frac{|00\rangle + |11\rangle}{\sqrt{2}}$ and $|B_{01}\rangle = \frac{|00\rangle - |11\rangle}{\sqrt{2}}$. We thus have that the resource state can be expressed as $|\Lambda \rangle= 0.5 |B_{00}\rangle + 0.5 |B_{01}\rangle$, which is a separable state. As a channel, this corresponds to Kraus operators  $K_m=\frac{1}{\sqrt{2}}Z^m$. From Proposition \ref{prop: diag}, we have that the teleported output state is now be given by $0.5 \sum_s X^s Z^t Z^m \rho Z^m Z^t X^s$. Thus, enough entanglement has been left in the resource state in order to perform a non-trivial task. 
\end{example}
Thus, the effect of a having access to a noisy resource state can be observed in the output of the computation, which crucially will depend on the noise channel that occurred. This thus motivates our study of one-dimensional MBQC within the context of a noisy resource state, and in the next section, we introduce a framework which will allow us to map the effect of local noise channels acting on the state onto the computation. 
\section{Noisy MBQC} \label{sec:noisyMBQC}
One-dimensional MBQC consists of adaptive single qubit measurements in either the computational basis or onto the equatorial plane of the Bloch sphere to be performed on a one-dimensional cluster state, as illustrated in Fig. \ref{fig:MBQCcirc}. In the following, our goal is to study the effect of local, single qubit noise channels acting on the physical state and determine their effect on the computation. In order to do so, we choose to break down MBQC into fundamental ``building blocks'', which can then be concatenated thus forming the complete computation. Then, the effect of noise channels on the individual building block is studied, and thus their effect on the output state is determined. Finally, we will argue that by concatenating such noisy blocks, we can model the effect of noise on the entire computation, thus understanding how noise affecting the physical state gets mapped to noise on the computation.   
\begin{figure}[h!]
	\centering 
	\mbox{
		\Qcircuit @C=1.5em @R=1.9em {
			\lstick{|+\rangle }	& \ctrl{1} 	 & \qw    	\\
			\lstick{|+\rangle }	& \ctrl{-1} &\ctrl{1} 	 & \qw \\
			\lstick{|+\rangle }	& \qw 		 & \ctrl{-1} & \ctrl{1} & \qw \\
			\lstick{|+\rangle }	& \qw  		 & \qw 		& \ctrl{-1} & \ctrl{1} & \qw  \\
			\lstick{|+\rangle }	& \qw  		 & \qw 		& \qw        & \ctrl{-1} & \qw 
		} }
		\caption[A One-dimensional cluster state.]{\emph{A one-dimensional cluster state.} Circuit of a one-dimensional $5$ qubit cluster state. Each qubit is initialised in the $|+\rangle$ state, and then entangled via a controlled-$Z$ gate with its neighbour.} 
		\label{fig:MBQCcirc}
	\end{figure}
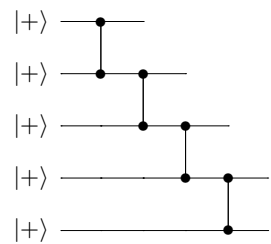
	\subsection{The Noiseless Model via The Building Block}
	
\emph{---The Building Block---} The fundamental building block of MBQC in one dimension is a two qubit state, where the first qubit is in an arbitrary mixed state $\rho$ and where the second qubit is a pure qubit initialised in the $|+\rangle$ state. These are then entangled via a controlled-$Z$ operator before a final measurement of either the observable $Z$ or  $R_z(\phi) X R_z(-\phi)$, denoted by $\phi$, is performed on the first qubit, yielding outcome $k$, see Fig. \ref{fig:equMeas}. This circuit results in the state $\rho$ being teleported onto the second qubit, with the added application of a unitary operator --- dependent on the measurement basis --- and a potential Pauli by-product operator due to the randomness of measurement outcomes.
	
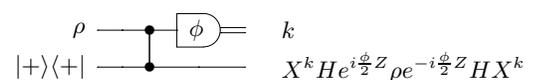
\begin{figure}[h!]
		\mbox{
			\Qcircuit @C=1em @R=.7em {
				\lstick{\rho}                           	& \qw    & \ctrl{1}   &  \measureD{\phi }      & \cw     & \rstick{k}   \\
				\lstick{|+\rangle \langle + |}     & \qw    & \ctrl{-1} &  \qw            & \qw  & \rstick{  X^k H e^{i \frac{\phi}{2}Z} \rho  e^{-i \frac{\phi}{2}Z} H X^k}
			} }
			\caption[Building Block]{\emph{The MBQC Building Block.}  The two qubits are entangled via a controlled-$Z$, and the first qubit is measured in the equatorial plane of the Bloch sphere yielding measurement outcome $k$.} 
			\label{fig:equMeas}
\end{figure}

\emph{---Block Notation---} Depending on whether the measurement was of the observable $Z$ or $R_z(\phi) X R_z(-\phi)$, the output is given respectively either by $\varepsilon_{z, k}(\rho)=  \frac{1}{2} Z^k \rho Z^k$ or $\varepsilon_{\phi, k} (\rho)=\frac{1}{2}  X^k H e^{i \frac{\phi}{2}Z} \rho  e^{-i \frac{\phi}{2}Z} H X^k$. We introduce the following shorthand notation: let  $\varepsilon_{i,k}(\rho)$ denote the output of the computation, where $i \in \{Z,  \phi \}$ refers to the observable measured and $k$ to the measurement outcome. This is illustrated in Fig. \ref{fig:genBlock}, with the effective implemented quantum channel in Fig.\ref{fig:channel}. 
\begin{figure}[h!]
				\centering 
				\mbox{
					\Qcircuit @C=1em @R=.7em {
						\lstick{\rho}	& \qw & \ctrl{1}  &  \measureD{i } & \cw   & \rstick{k}   \\
						\lstick{|+\rangle \langle + |}   	 & \qw    & \ctrl{-1} &   \qw  & \qw & \rstick{ \varepsilon_{i,k} (\rho ) }
					} }
					\caption[Building Block]{\emph{The MBQC Building Block.}  The two qubits are entangled, and the first qubit is measured in either the computational basis or in the equatorial plane of the Bloch sphere.} 
					\label{fig:genBlock}
\end{figure}
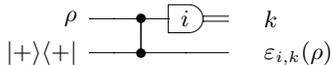
				
\begin{figure}[h!] 
					\centering %
					\mbox{
						\Qcircuit @C=1em @R=.7em {
							\lstick{\rho }	& \qw & \gate{\varepsilon_{i, k}}  & \qw  & \qw &\rstick{\varepsilon_{i,k} (\rho)}   
						} }
						\caption[Building Bloch Channel]{\emph{Channel Building Block} The building block effectively implements a quantum channel $\varepsilon_{i,k}$ on the input state $\rho$ where $i$ refers the measurement bases and $k$ the measurement outcome.} 
						\label{fig:channel}
\end{figure}
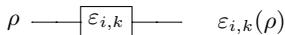		
\emph{---Effective computation---}The measurement thus implements a unitary operator on the state $\rho$ described by the superoperator $\varepsilon_{i,m}$. The output state can then be used as the input to a new block, a process which will be repeated until the computation terminates. If it consists of $L$ such blocks, then the final output state is given by $\rho'=\varepsilon_{i_L, k_L} \circ \ldots \varepsilon_{i_2, k_2} \circ \varepsilon_{i_1, k_1} (\rho )$, where $\circ$ denotes the composition of superoperators, i.e. $\alpha \circ \varepsilon (\rho)= \alpha (\varepsilon (\rho))$, and where $i_j \in \{z, \phi \}$, $k_j \in \{0,1\}$for $j= 1, \ldots, L$. We note that the measurement bases are still dependent on previous measurement outcomes, although we have for clarity chosen not to explicitly state this dependence. 
					
We have thus now decomposed one-dimensional MBQC as a sequence of building blocks acting on an input state $\rho$. Henceforth, we shall exclusively consider the building block and investigate how noise acting on the physical states is mapped to noise acting on the computation. But first, we consider a couple of simple examples in order to develop an intuition as to how noise propagates in this model. 
\subsection{General Local Noise}
We now introduce a framework which will allow for the study of arbitrary noise acting on single qubits, which provides a systematic method for mapping errors acting on the circuit to errors acting on the computation. In order to do so, we will consider how errors acting at various locations of the fundamental building block are transmitted.
									
Any quantum state $\rho$ can be decomposed in the computational basis as $\rho=\sum \alpha_{i_0 j_0} |i_0 \rangle \langle j_0|$, and is thus a linear combination of basis elements $|i_0 \rangle \langle j_0|$. The Choi matrix $C_\varepsilon$ of a CPTP map $\varepsilon$ stores the transformed basis states $\varepsilon (|i_0 \rangle \langle j_0|)$ of a quantum channel $\varepsilon$, in a similar way to matrices which store transformed basis vectors. By linearity, the Choi matrix thus provides us with a way of computing the output state of a quantum channel. Thus, we henceforth limit our study to that of inputs of the type $|i_0 \rangle \langle j_0|$. 
									
We consider a local single qubit noise model, represented by a CPTP map having a Kraus representation $\alpha (\rho ) = \sum_m K_m \rho K_m^\dag$. Each Kraus operator is a matrix of dimension two, which can thus be decomposed in the Pauli basis $K_m=\sum_{g,h} a_{gh}\sigma_{gh}$, where $\sigma_{gh}=i^{gh}X^gZ^h$. Each noise channel $\alpha$ may act at one or more of the four possible locations depicted in Fig.\ref{fig:arbNoise}.
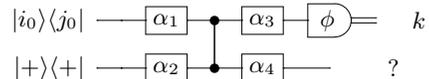
\begin{figure}[h!]
\centering 
\mbox{
\Qcircuit @C=1em @R=.7em {
\lstick{|i_0 \rangle \langle j_0|}	& \qw &\gate{\alpha_1}& \ctrl{1}  & \gate{\alpha_3} &\measureD{\phi} & \cw   & \rstick{k}   \\
\lstick{|+\rangle \langle + |}   	 & \qw   &\gate{\alpha_2} & \ctrl{-1}   &\gate{\alpha_4}& \qw & \rstick{?}
} }
\caption[General noise on the building block.]{\emph{General noise on the building block.} An input $|i_0 \rangle \langle j_0|$ is entangled with a $|+\rangle \langle+|$ before the first qubit is measured in the equatorial plane of the Bloch sphere. The noise channels $\alpha_i$ may occur at four possible locations, thus disrupting the computation.} 
\label{fig:arbNoise}
\end{figure}
Finally, we note that as $Z$ measurements have the effect of destroying the entanglement present between the measured qubit and its neighbours, we shall now exclusively focus on measurements performed onto the equatorial plane of the Bloch sphere, and write $\varepsilon_k$ instead of $\varepsilon_{\phi, k}$. 
										
First, we introduce two lemmas which treat the two trivial cases of a noise channel acting directly on the input or output states. If the noise acts on the input, then it is this mixed state which is teleported. 
\begin{figure}[h!] 
	\centering 
	\mbox{
		\Qcircuit @C=1em @R=.7em {
			\lstick{|i_0\rangle \langle j_0 |}	& \qw & \gate{\alpha_1} & \ctrl{1}  &  \measureD{\phi } & \cw   & \rstick{k}   \\
			\lstick{|+\rangle \langle + |}   	 & \qw  & \qw  & \ctrl{-1} &   \qw  & \rstick{ \varepsilon_k \circ \alpha_1 (|i_0 \rangle \langle j_0|) }
		} }
		\caption[Noise acting on the input state.]{\emph{Noise acting on the input state.}The noise channel $\alpha_1$ acts directly on the input $|i_0 \rangle \langle j_0|$, which is thus directly transmitted, resulting in the output $ \varepsilon_k \circ \alpha_1 (|i_0 \rangle \langle j_0|)$.} 
		\label{fig:alpha1}
	\end{figure}
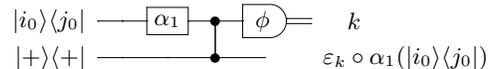
\begin{lemma} \label{lemma: noisy Input}
Let $\alpha_1$ be the single qubit noise channel acting at location $1$, as shown in Fig.\ref{fig:alpha1}. Then, the output is given by the channel 
\begin{equation}
\varepsilon_k \circ \alpha_1
\end{equation}. 
\end{lemma}

\begin{proof}
This is equivalent to the noise channel $\alpha_1$ acting on the input state itself, and thus the teleportation protocol is directly applied to $\alpha_1(|i_0 \rangle \langle j_0|)$, thus resulting in the state $\varepsilon_k \circ \alpha_1 (|i_0 \rangle \langle j_0|)$ being transmitted. 
\end{proof}
Next, the following lemma formalises the trivial case when the noise channel acts on the output state. 
\begin{figure}[h!] 
	\centering 
	\mbox{
		\Qcircuit @C=1em @R=.7em {
			\lstick{|i_0\rangle \langle j_0 |}	& \qw   & \ctrl{1}   & \qw &  \measureD{\phi } & \cw   & \rstick{k}   \\
			\lstick{|+\rangle \langle + |}   	 & \qw   & \ctrl{-1} & \gate{\alpha_4} &  \qw  & \rstick{\alpha_4 \circ \varepsilon_k(|i_0 \rangle \langle j_0|) }
		} }
		\caption[Noise acting on the output state.]{\emph{Noise acting on the output state.} The noise channel $\alpha_4$ acts directly on the output thus resulting in the output $\alpha_4 \circ \varepsilon_k(|i_0 \rangle \langle j_0|)$.} 
		\label{fig:alpha4}
	\end{figure}
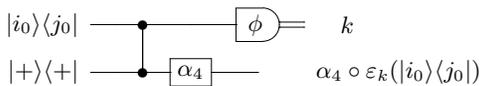

\begin{lemma} \label{lemma: noisy output}
Let $\alpha_4$ be the single qubit noise channel acting at location $4$, as shown in Fig.\ref{fig:alpha4}. Then, the output is given by the channel 
\begin{equation}
\alpha_4 \circ \varepsilon_k.
\end{equation}
\end{lemma}

Next, we generalise the previously introduced noisy cluster state and noisy measurement examples to arbitrary CPTP maps. First, when the noise channel is applied onto the second qubit---the resource---just before the states are entangled, this corresponds to a noisy cluster state. MBQC relies on the preparation of a pure cluster state which will serve as a reservoir of correlations to be used by the subsequent computation. As we previously saw, noisy resource states are not useless, and thus perhaps could lead to an interesting computation being performed. Here, we model this noisy cluster state as a local noise channel acting on the $|+\rangle$ input states. 
\begin{figure}[h!] 
	\centering 
	\mbox{
		\Qcircuit @C=1em @R=.7em {
			\lstick{|i_0\rangle \langle j_0 |}	& \qw & \qw & \ctrl{1}  &  \measureD{\phi } & \cw   & \rstick{m}   \\
			\lstick{|+\rangle \langle + |}   	 & \qw  & \gate{\alpha_2}  & \ctrl{-1} &   \qw  & \rstick{\tilde{\alpha_2} \circ  \varepsilon_k (|i_0 \rangle \langle j_0|) }
		} }
		\caption[Noise acting on the resource state.]{\emph{Noise acting on the resource state.} The noise channel $\alpha_2$ acts directly on the resource, thus resulting in the output $\tilde{\alpha_2} \circ  \varepsilon_k (|i_0 \rangle \langle j_0|)$.} 
		\label{fig:alpha2}														
	\end{figure}
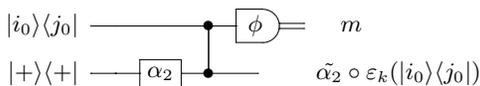
\begin{lemma} \label{lemma: noisy resource}
Let $\alpha_2$ be the single qubit noise channel acting at location $2$, as shown in Fig.\ref{fig:alpha2} and which is expressed in Kraus form $\alpha_2(\rho)=\sum_m K_m \rho K_m^\dag$. If the Kraus operators are decomposed in the Pauli basis $K_m=\sum_{gh}a_{gh}^{(m)} \sigma_{gh}$, then the output is given by: 
\begin{equation}
\tilde{\alpha_2} \circ \varepsilon_k, 
\end{equation}
 where  the modified channel $\tilde{\alpha}_2(\rho)=\sum_m \tilde{K}_m \rho \tilde{K}_m^\dag$ has Kraus operators given by:
\begin{equation} \label{equ: noise on cluster}
\tilde{K}_m=\sum_u \tilde{a}_{u0}^{(m)} Z^{(m)u}, 
\end{equation}
 with $\tilde{a}_{u0}= \sum_v  (-i)^{uv} a_{uv}^{(m)}$. 
\end{lemma}	
\begin{proof}
The proof is given in the appendix \ref{app: noisy cluster}. The Kraus operators $K_m$ are decomposed in the Pauli basis, and their effect on the resource state $|+\rangle$ is studied, thus giving the output. 
\end{proof}												
The results from lemma \ref{lemma: noisy resource} are summarised in table \ref{tab:noisy cluster}, which illustrates how each Pauli basis element is mapped to a Pauli element in order to form the final noise channel acting on the output. 
\begin{table}[h]
	\centering % centering table
	\begin{tabular}{c c c c c } % creating 10 columns
		\hline\hline % inserting double-line
				Coefficient 		  & $a_{00}$	   & $a_{01}$ & $a_{10}$ & $a_{11}$  \\[1ex]

		\hline % inserts single-line
		% Entering 1st row
		Initial operator 	& $\mathbb{I}$ & $X$ & $Z$ & $-iZX$ \\ [0.5ex]
		% Entering 2nd row
		Final operator		  & $\mathbb{I}$	   & $\mathbb{I}$ & $Z$ & $-iZ$  \\[1ex]
		% [1ex] adds vertical space
		\hline % inserts single-line
	\end{tabular}
	\caption[Mapping of Pauli basis elements for a noisy cluster.]{\emph{Mapping of Pauli basis elements for a noisy cluster.} The initial channel $\alpha_2$ is given by Kraus operators which can be decomposed in the Pauli bases $\sigma_{gh}$ with coefficient $a_{gh}$. The final channel $\tilde{\alpha}_2$ is given in terms of Pauli operators $\mathbb{I}$ and $Z$, with coefficients determined by the mapping.} 
	\label{tab:noisy cluster}
\end{table}				

Finally, we consider the case of a noisy measurement, which is modelled by applying a noise channel to the first qubit just before the measurement is performed. Here, we can exploit the fact that the measurement basis is known. Indeed, a unitary operator $U$ followed by a projective measurement in the basis $|v_k \rangle$ can alternatively be viewed as a measurement onto the state $U^\dag |v_k \rangle$. Here, the measurements are onto $e^{-i \frac{\phi}{2}Z}Z^k|+\rangle$, and we thus choose to decompose the Kraus operators in the rotated Pauli basis $U \sigma_{gh}U^\dag$, with $U=e^{-i \frac{\phi}{2}Z}$. Hence, each Kraus operator can now be expressed as $K_m= \sum a_{gh} U \sigma_{gh} U^\dag$, where $a_{gh}= \text{Tr}(K_m U\sigma_{gh}U^\dag)$, and where we now define $\sigma_{gh} =(-i)^{gh} Z^g X^h$. Thus, we can study the effect of a noisy measurement and find that: 
\begin{figure}[h!] 
	\centering 
	\mbox{
		\Qcircuit @C=1em @R=.7em {
			\lstick{|i_0\rangle \langle j_0 |}	& \qw  & \ctrl{1}   &  \gate{\alpha_3 } & \measureD{\phi } & \cw   & \rstick{k}   \\
			\lstick{|+\rangle \langle + |}   	 & \qw  & \ctrl{-1} &   \qw 	& \qw 		  & \rstick{ \tilde{\alpha}_{3, k} \circ \varepsilon(|i_0 \rangle \langle j_0|) }
		} }
		\caption[Noisy Measurement.]{\emph{Noisy Measurement.} The noise channel $\alpha_3$ acts on the first qubit just before the measurement in the equatorial plane. This results in the output given by $\tilde{\alpha}_{3, k} \circ \varepsilon(|i_0 \rangle \langle j_0|)$, where the noise now also depends on the measurement outcome $k$.}
		\label{fig:noisy measurement}
\end{figure}
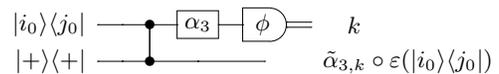

\begin{lemma} \label{lemma: noisy meas}
Let $\alpha_3$ be the single qubit noise channel acting at location $3$, as shown in Fig.\ref{fig:noisy measurement} and which is expressed in Kraus form $\alpha_3(\rho)=\sum_m K_m \rho K_m^\dag$. If the Kraus operators are decomposed in the Pauli basis $K_m=\sum_{gh} a_{gh}^{(m)}e^{-i\frac{\phi}{2} Z} \sigma_{gh} e^{i\frac{\phi}{2} Z}$, then the output is given by: 
\begin{equation}
\varepsilon_k  \circ \tilde{\alpha}_{3,k}
\end{equation}
 where the channel $\tilde{\alpha_{3}}^k$ has Kraus operators defined as: 
\begin{equation}\label{equ: noise on meas}
\tilde{K}_{m}=\sum_{u} \tilde{a}_{uk}^{(m)}   Z^u, 
\end{equation}
with coefficients $\tilde{a}_{uk}^{(m)}= \sum_{v} (i)^{uv}(-1)^{kv}a_{uv}^{(m)}$ and $k$ denoting the measurement outcome. 
\end{lemma}			

\begin{proof}
The proof relies on decomposing each of the Kraus operators in the rotated Pauli basis and studying its effect on the computation, given in appendix \ref{app: rotated}. 
\end{proof}
This mapping is summarised in table \ref{tab:noisy meas}, where each Pauli basis element is mapped to a Pauli element, with certain coefficients now depending on measurement outcome $k$. 
\begin{table}[h]
	\centering % centering table
	\begin{tabular}{c c c c c } % creating 10 columns
		\hline\hline % inserting double-line
		Coefficient 		  & $a_{00}$	   & $a_{01}$ & $a_{10}$ & $a_{11}$  \\[1ex]
		
		\hline % inserts single-line
		% Entering 1st row
		Initial  	& $\mathbb{I}$ & $e^{-i\frac{\phi}{2} Z} Ze^{i\frac{\phi}{2} Z} $ & $e^{-i\frac{\phi}{2} Z} Xe^{i\frac{\phi}{2} Z} $ & $ie^{-i\frac{\phi}{2} Z} XZe^{i\frac{\phi}{2} Z} $ \\ [0.5ex]
		% Entering 2nd row
		Final 		  & $\mathbb{I}$	   & $Z$ & $(-1)^k \mathbb{I}$ & $i(-1)^k Z$  \\[1ex]
		% [1ex] adds vertical space
		\hline % inserts single-line
	\end{tabular}
	\caption[Mapping of Pauli basis elements for a noisy measurement.]{\emph{Mapping of Pauli basis elements for a noisy measurement.} The initial channel $\alpha_3$ is given by Kraus operators which can be decomposed in the rotated Pauli bases $e^{-i\frac{\phi}{2} Z}\sigma_{gh}e^{i\frac{\phi}{2} Z}$ with coefficient $a_{gh}$. The final channel $\tilde{\alpha}_3$ is given in terms of Pauli operators $\mathbb{I}$ and $Z$, with coefficients determined by the mapping.} 
	\label{tab:noisy meas}
\end{table}				
														
We have thus studied the mapping of noise onto computation, when a local noise channel acts in one of the four possible locations of the MBQC building block. Finally, we consider the case whereby noise acts simultaneously at all four locations of the building block. By combing the four previous lemmas, we obtain the following theorem, mapping noise on the state to noise on the computation:
\begin{theorem} \label{theorem: noise block}
Let $\alpha_i$, for $i=1,2,3,4$ be noise channels acting at all four possible locations of the building block, as depicted in Fig. \ref{fig:arbNoise}. Each channel is a single qubit CPTP map which may be expressed as a Kraus decomposition. Then, upon input $|i_0\rangle \langle j_0|$, the output is given by the channel: 
\begin{equation}
\alpha_4 \circ \tilde{\alpha}_2 \circ \varepsilon_k \circ \tilde{\alpha}_{3,k} \circ \alpha_1,  
\end{equation}
where the new channels $\tilde{\alpha}_2$ and $\tilde{\alpha}_{3,k}$ are defined as follows and where $\varepsilon_k(\rho)=\frac{1}{2}  X^k H e^{i \frac{\phi}{2}Z} \rho  e^{-i \frac{\phi}{2}Z} H X^k$.

If the original channel $\alpha_3$ can be expressed as a Kraus decomposition $\alpha_3(\rho)=\sum_m K_m \rho K_m^\dag$, and each Kraus operator decomposed in the rotated Pauli basis $K_m=\sum_{gh} a_{gh}^{(m)}e^{-i\frac{\phi}{2} Z} \sigma_{gh} e^{i\frac{\phi}{2} Z}$, then the channel $\tilde{\alpha_{3}}^k$ will have Kraus operators given by: 
\begin{equation}\label{equ: noise on meas}
\tilde{K}_{m}=\sum_{u} \tilde{a}_{uk}^{(m)}   Z^u, 
\end{equation}
with coefficients $\tilde{a}_{uk}^{(m)}= \sum_{v} (i)^{uv}(-1)^{kv}a_{uv}^{(m)}$ and $k$ denoting the measurement outcome. If the original channel $\alpha_2$ is expressed as a Kraus decomposition $\alpha_2(\rho)=\sum_m K_m \rho K_m^\dag$, then the channel $\tilde{\alpha}_2(\rho)$ will have Kraus operators given by:
\begin{equation} \label{equ: noise on cluster}
\tilde{K}_m=\sum_u \tilde{b}_{u0}^{(m)} Z^{(m)u}, 
\end{equation}
and with coefficients $\tilde{b}_{u0}= \sum_v  (-i)^{uv} b_{uv}^{(m)}$. 
\end{theorem}
\begin{proof}
The theorem can be proven by combining lemmas \ref{lemma: noisy Input}, \ref{lemma: noisy output}, \ref{lemma: noisy resource} and \ref{lemma: noisy meas}, and can be found in appendix \ref{app: theorem}.
\end{proof}
\subsection{Examples}
Thus, we next revisit the previous examples of the bit-flip and phase-flip channels acting on either the cluster state or the measurement, before considering a more general noise model. First, we consider a noisy cluster state, which may easily be studied using lemma \ref{lemma: noisy resource} and table \ref{tab:noisy cluster}. 
\begin{example}
When the bit-flip channel is applied, the Kraus operators $K_s = \sqrt{p_s} X^s$, $s=0, 1$ are mapped to $\tilde{K}_s=\sqrt{p_s} \mathbb{I}$, thus resulting in the identity channel: the output state $\varepsilon_k (\rho)$ has been unaffected by the noise channel.
\end{example}
\begin{example}
When the phase-flip channel is applied, the Kraus operators $K_s = \sqrt{p_s} Z^s$, $s=0, 1$ remain unchanged under the mapping $\tilde{K}_s=\sqrt{p_s} Z^s$, and thus the phase-flip channel has been applied to the output $\sum_s p_s Z^s\varepsilon (\rho)Z^s$.
\end{example}

Next, we consider the case of a noisy measurement in the $X$ basis, which may easily be studied using lemma \ref{lemma: noisy meas} and table \ref{tab:noisy meas}, which reduces to table \ref{tab:X meas} for the case $\phi=0$. Once again, we consider the effect of the bit-and phase-flip channels on the output. 
\begin{table}[h]
	\centering % centering table
	\begin{tabular}{c c c c c } % creating 10 columns
		\hline\hline % inserting double-line
		Coefficient 		  & $a_{00}$	   & $a_{01}$ & $a_{10}$ & $a_{11}$  \\[1ex]
		
		\hline % inserts single-line
		% Entering 1st row
		Initial  	& $\mathbb{I}$ & $Z $ & $ X$ & $iXZ$ \\ [0.5ex]
		% Entering 2nd row
		Final 		  & $\mathbb{I}$	   & $Z$ & $(-1)^k \mathbb{I}$ & $i(-1)^k Z$  \\[1ex]
		% [1ex] adds vertical space
		\hline % inserts single-line
	\end{tabular}
	\caption[Mapping of Pauli basis elements for a noisy $X$ measurement.]{\emph{Mapping of Pauli basis elements for a noisy $X$ measurement.} Here, we consider the special case of lemma \ref{lemma: noisy meas}, where the measurement is in the $X$ basis and thus $\phi=0$. The initial channel $\alpha_3$ is given by Kraus operators which can be decomposed in the Pauli bases $\sigma_{gh}$ with coefficient $a_{gh}$. The final channel $\tilde{\alpha}_3$ is given in terms of Pauli operators $\mathbb{I}$ and $Z$, with coefficients determined by the mapping.} 
	\label{tab:X meas}
\end{table}		
\begin{example}
When the bit-flip channel is applied, the Kraus operators $K_s = \sqrt{p_s} X^s$, $s=0, 1$ are mapped to $\tilde{K}_s=\sqrt{p_s} (-1)^{ks} \mathbb{I}$, thus resulting in the identity channel: the output state $\varepsilon_k (\rho)$ has been unaffected by the noise.

\end{example}
\begin{example}
When the phase-flip channel is applied, the Kraus operators $K_s = \sqrt{p_s} Z^s$ are mapped to $\tilde{K}_s=\sqrt{p_s} Z^s$, thus resulting in the phase-flip channel acting on the output $\sum_s p_s Z^s\varepsilon_k (\rho)Z^s$.
\end{example}			
					
Next, having considered Pauli noise, we look at the more general case of unitary noise, whereby a unitary operator is applied to a physical qubit with probability $p$, thus disrupting and potentially destroying the computation. In the following, we shall study the effect of a Hadamard gate $H=\frac{1}{\sqrt{2}}(X+Z)$ on both the cluster and the measurement. 
\begin{table}[h]
	\centering % centering table
	\begin{tabular}{c c c c c } % creating 10 columns
		\hline\hline % inserting double-line
		Coefficient 		  & $0$	   & $\sqrt{\frac{p}{2}}$ & $\sqrt{\frac{p}{2}}$  & $0$  \\[1ex]
		
		\hline % inserts single-line
		% Entering 1st row
		Initial operator 	& $\mathbb{I}$ & $X$ & $Z$ & $-iZX$ \\ [0.5ex]
		% Entering 2nd row
		Final operator		  & $\mathbb{I}$	   & $\mathbb{I}$ & $Z$ & $-iZ$  \\[1ex]
		% [1ex] adds vertical space
		\hline % inserts single-line
	\end{tabular}
	\caption[A Hadamard acting on the cluster. ]{\emph{A Hadamard acting on the cluster.} The initial channel $\alpha_2$ is given by Kraus operators which can be decomposed in the Pauli bases $\sigma_{gh}$ with coefficient $a_{gh}$. The final channel $\tilde{\alpha}_3$ is given in terms of Pauli operators $\mathbb{I}$ and $Z$, with coefficients determined by the mapping.} 
	\label{tab:noisy cluster}
\end{table}	
\begin{example} \label{ex: H cluster}
When a Hadamard gate is applied with probability $p_1$, the Kraus operators of the corresponding channel are $K_s=\sqrt{p_s} H^s$, $s=0, 1$. From table \ref{tab:noisy cluster} which are mapped to $\tilde{K}_0=K_0$ and $\tilde{K}_1=\sqrt{2p_1} |0\rangle \langle 0|$ thereby destroying the computation.

\end{example}

\begin{table}[h]
	\centering % centering table
	\begin{tabular}{c c c c c } % creating 10 columns
		\hline\hline % inserting double-line
		Coefficient 		  & $0$	   & $\frac{\sqrt{p}}{\sqrt{2}}$ & $\frac{\sqrt{p}}{\sqrt{2}}$ & $0$  \\[1ex]
		
		\hline % inserts single-line
		% Entering 1st row
		Initial  	& $\mathbb{I}$ & $ Z $ & $ X $ & $i ZX  $ \\ [0.5ex]
		% Entering 2nd row
		Final 		  & $\mathbb{I}$	   & $Z$ & $(-1)^k \mathbb{I}$ & $i(-1)^k Z$  \\[1ex]
		% [1ex] adds vertical space
		\hline % inserts single-line
	\end{tabular}
	\caption[A Hadamard before the measurement.]{\emph{A Hadamard before the measurement.} The initial channel $\alpha_3$ is given by Kraus operators which can be decomposed in the rotated Pauli bases $\sigma_{gh}$ with coefficient $a_{gh}$. The final channel $\tilde{\alpha}_3$ is given in terms of Pauli operators $\mathbb{I}$ and $Z$, with coefficients determined by the mapping.} 
	\label{tab:noisy measother}
\end{table}	
\begin{example} \label{ex: H meas}
Here, with probability $p$ a Hadamard gate is applied to the second qubit before the control-$Z$ and a measurement in the $X$ basis is performed. This corresponds to Kraus operators $K_1=\sqrt{1-p} \mathbb{I}$ and $K_2=\sqrt{p} H$. From table \ref{tab:noisy meas}, the new modified Kraus operators can be immediately computed, and we find that: $\tilde{K}_1=K_1$ and $\tilde{K}_2=\sqrt{2p}|k\rangle \langle k|$
\end{example}										
Thus, we have been able to use the framework derived in order to map physical errors acting on qubits to logical errors acting throughout the computation. Next, we consider the framework of MPS which is a natural framework for the study of one-dimensional MBQC. Here, the effect of measurements on the physical state can easily be observed as logical operators acting on the correlation space. This provides us with motivation to investigate how a mixed state generalisation of MPS could be used in order to study noisy MBQC, and understand the mapping between the noise channels acting on the state and the information  processed. 
\section{Matrix Product Operators} \label{sec: MPO}	
Matrix Product Operators (MPOs) are an extension of the formalism of MPS to mixed states. In order to develop an intuition as to what this might look like, we first consider the example of a one-dimensional cluster $|\psi \rangle$ expressed as an MPS:
\begin{equation}
|\psi \rangle = \sum_i \langle i_n |A[i_{n-1}] \ldots A[i_1] |+\rangle |i_1 \ldots i_n \rangle, 
\end{equation}
where $A[k]=H|k\rangle \langle k|$, $k=0,1$ are logical operators acting within the correlation space. When a physical qubit is measured, its associated logical operator is updated and mapped to a new logical operator dependent on the measurement basis and outcome. Thus, as measurements are implemented, the evolution of the computation can be observed within the correlation space, with operators proportional to unitary operators acting on the logical $|+\rangle$ state. Hence, this is has been a natural framework for the study of one-dimensional MBQC. 

In order to generalise the MPS framework to mixed states, we first consider the rank one density operator of a pure cluster state: 
\begin{align}
|\psi \rangle  \langle \psi |=& \sum_i \sum_j \langle i_n |A[i_{n-1}] \ldots A[i_1] |+ \rangle \langle + | A[j_1] \notag \\
& \ldots A[j_{n-1}] |j_n \rangle |i_1 \ldots  i_n \rangle \langle j_1 \dots j_n|.
\end{align}
Previously, we thought of the logical operators $A[i]$ as acting on the logical $|+\rangle$ state, in analogy with unitary operators acting upon quantum states. Here, we wish to consider mixed state evolution, which is typically represented by a superoperator acting on a density matrix. We thus propose to now interpret the MPO representation as that of a map acting upon an input $|+\rangle \langle +|$, and we thus write: 
\begin{align}
|\psi \rangle  \langle \psi |=& \sum_i \sum_j  \sigma^{[i_n, j_n]} \circ \varepsilon^{[i_n, j_n]} \ldots \varepsilon^{[i_1, j_1]} (|+\rangle \langle+|) \notag \\
& |i_1 \ldots i_n \rangle \langle j_1 \ldots  j_n|,
\end{align}
where we have introduced the logical superoperator $\varepsilon^{[i_k, j_k]}$ acting on qubit $k$, which is defined as:
\begin{equation}
\varepsilon^{[i_k, j_k]}(\rho)=A[i_k] \rho  A[j_k], 
\end{equation}
an expression similar to the Kraus decomposition of a CPTP map, and where the boundary conditions are given by: 
\begin{equation}
\sigma^{[i_n, j_n]} (\rho)=\langle i_n |\rho |j_n \rangle. 
\end{equation}

This intuition is formalised in lemma \ref{lemma:MPO}, where an expression for an MPO is given in terms of logical superoperators acting on an input within the correlation space. To our knowledge, similar MPO representations were first studied in \cite{pirvu2010matrix}. The key idea is that the initial mixed state is first purified before then, successive Schmidt decompositions are applied to the state and finally, the auxiliary qubits are traced out. Thus, the coefficients of the final state will be expressed as the composition of logical superoperator acting on a rank one density matrix. 
\begin{lemma}\label{lemma:MPO}
	An $n$-qubit mixed state $\rho$ can be expressed as an MPO: 
	\begin{align}
	\rho=& \sum_{i, j}\sigma^{[i_n,j_n]} \circ \varepsilon^{[i_{n-1}, j_{n-1}]}\circ \dots  \circ \varepsilon^{[i_{2}, j_{2}]}(\rho[i_1, j_1])  \notag \\
	&  |i_1 \ldots i_n \rangle \langle j_1 \ldots j_n |, 
	\end{align}
	where $\varepsilon^{[i_{k}, j_{k}]}(\rho)=\sum_s A[i_k, s] \rho A^\dag[j_k, s]$ denote the logical superoperator associated with qubit $k$ and where the boundary conditions are given by $\sigma^{[i_n,j_n]} (\rho)= \sum_s v^\dag[i, s] \rho v[j, s]$. $A[i_k, s]$ correspond to matrices of dimension $\xi$ and the $v[j_k, s]$ to vectors, both of which depend on the state of the physical qubit $k$ and an index $s$. 
\end{lemma}
The proof of lemma \ref{lemma:MPO} is given in appendix \ref{app: MPO}. When a single qubit measurement is performed, the associated logical superoperator is evolved, which is then given by the following proposition.  
\begin{proposition}\label{lemma:MPO meas}
	When the $k^{th}$ physical qubit is measured in an orthonormal basis $|v_m\rangle$, the logical superoperator $\varepsilon^{[i_k, j_k]}$ evolves to $\varepsilon^{[v_m, v_m]} (\rho)=\sum_s A[v_m,s]\rho A^\dag[v_m,s]$, where $A[v_m,s]=\sum_i \langle v_m|i\rangle A[i,s]$.
\end{proposition}
The proof simply relies on applying the measurement and studying its effect on the relevant operators. 

Thus, having derived and expression for the MPO of a mixed state, we now consider a couple of examples where this representation can be applied. The first example is the density matrix corresponding to a pure cluster state. 
\begin{example}
	The MPO representation of a cluster state is given by:  
	\begin{align}
	\rho=&\sum_{i, j} \sigma^{[i_n, j_n]} \varepsilon^{[i_{n-1}, j_{n-1}]} \circ \ldots \circ \varepsilon^{[i_2, j_2]} (|+\rangle \langle +|) \notag \\
	& |i_1 \ldots i_n\rangle \langle j_1 \ldots j_n|,
	\end{align}
	where $\varepsilon^{[i,j]}(\rho)=A[i] \rho A[j]$, $A[k]=H|k\rangle \langle k|$ and $\sigma^{[i,j]}(\rho)=\langle i |\rho|j\rangle$. 
\end{example}
Next, we consider the case of the maximally mixed state on $n$ qubits, $\frac{1}{2^n}\mathbb{I}_n$. This can be easily obtained by first constructing the all zero state, and then applying the channel $K_s=\frac{1}{\sqrt{2}}X^s$ to every qubit. 
\begin{example}
	The MPO representation of the maximally mixed state on $n$ qubits is given by
	\begin{align}
	\frac{1}{2^n}\mathbb{I}_n= &\sum_{i, j} \sigma^{[i_n, j_n]} \circ \varepsilon^{[i_{n-1}, j_{n-1}]} \circ \dots \circ \varepsilon^{[i_1, j_1]}(|0\rangle \langle 0|) \notag \\
	&|i_1 \ldots i_n\rangle \langle j_1 \ldots j_n|
	\end{align}
	where $\varepsilon^{[i,j]}(\rho)=\frac{1}{2}\sum_s A[i]X^s \rho X^s A[j]$, $ \sigma^{[i_n, j_n]}(\rho)=\frac{1}{2}\sum_{t} \langle i_n|X^t\rho X^t |j_n\rangle $, and $A[k]=|0\rangle \langle k|$. 
\end{example}
Finally, we consider the input state used in the one clean qubit model of computation, which is the state $|0\rangle \langle0| \otimes \frac{1}{2^n}\mathbb{I}$. 
\begin{example}
	The state $|0\rangle \langle0| \otimes \frac{\mathbb{I}}{2^n}$, which was the input state of the one clean qubit Model can be expressed as: 
	\begin{align}
	|0\rangle \langle 0| \otimes \frac{\mathbb{I}}{2^n}= & \sum_{\textbf{i}, \textbf{j}} \sigma^{[i_n, j_n]} \circ \varepsilon^[{i_{n-1}, j_{n-1}]} \circ \dots \circ \epsilon^{[i_0, j_0]}|0\rangle \langle 0| \notag \\
	& |i_i \ldots i_n \rangle \langle j_1 \ldots j_n |
	\end{align}
	where $\epsilon^{[i_0,j_0]}(\rho )=A[i_0]\rho A[j_0]^\dag$, $\varepsilon^{[i_k,j_k]}(\rho )=\frac{1}{2}\sum_{s_k} A[i_k]X^s \rho X^s A[j_k] ^\dag$,  $\sigma^{[i_n, j_n]}(\rho )=\frac{1}{2}\sum_{s_n} \langle i_n|X^t \rho X^t |j_n\rangle $, and $A[k]=|0\rangle \langle k|$. Here, we can note that the assumption of translational invariance does not hold. 
\end{example}
Thus, this shows us that the MPO framework can be applied to mixed states, and in the next section, we shall consider the case of noisy one-dimensional MBQC expressed in the MPO framework. 
\section{Noisy MBQC in the MPO framework} \label{sec:MPOMBQC}
In the following, we wish to consider one-dimensional MBQC in the case of a noisy resource state, and once again map the effect of physical errors to logical errors. In order to do so, we successively consider the effect of a Pauli operator, a unitary operator and finally a channel on a physical qubit and determine its effect on the associated logical operator.  

When a Pauli operator is applied to the state, the logical operator associated with the relevant qubit is transformed. Proposition \ref{prop: Pauli Noise} allows for this mapping to be determined:
\begin{proposition}\label{prop: Pauli Noise}
	If a Pauli operator $\sigma_{ab}$ is applied to the $j^{th}$ qubit, the logical operator $A[i_j]$ is mapped to the logical operator $\sigma_{0a}A[i_j]\sigma_{ab}$, as illustrated in table \ref{tab:Pauli on MPO}.
\end{proposition}
The proof is given in appendix \ref{app: pauli on MPO}, where we successively consider the effect of applying an individual Pauli operator to the physical state and determine its effect on the relevant logical operator. 
\begin{table}[h]
	\centering % centering table
	\begin{tabular}{c c c c c } % creating 10 columns
		\hline\hline % inserting double-line
		Pauli operator & $X$	   & $Z$ & $iXZ$   \\[1ex]
		
		\hline % inserts single-line
		% Entering 1st row
		Logical operator 	& $ZA[k]X$ & $A[k]Z$ & $iZA[k]XZ$  \\ [0.5ex]
		% Entering 2nd row
		\hline % inserts single-line
	\end{tabular}
	\caption[Mapping from Pauli errors to logical errors.]{\emph{Mapping from Pauli errors to logical errors.} A Pauli operator is applied to the $k^{th}$ physical qubit, which will cause its associated logical operator to evolve to a new logical operator.} 
	\label{tab:Pauli on MPO}
\end{table}	
Next, this allows for the effect of a unitary operator on the physical qubit to be determined in proposition \ref{prop:unitary}. Indeed, the unitary operator can now be decomposed in the Pauli basis, and thus its effect mapped onto the associated logical operator, bringing $A[k]$ to the transformed $\tilde{A}[k]$.
\begin{proposition}\label{prop:unitary}
	When a unitary operator $U=\sum_{g,h} u_{gh}\sigma_{gh}$ acts on the $j^{th}$ qubit, the associated logical operator evolves to: 
	\begin{equation}
	\tilde{A}[i_j] =\xi (A[i_j])= \sum_{g,h} u_{gh}\sigma_{0g}A[i_j] \sigma_{gh}. 
	\end{equation}
	This defines a channel $\xi$ which maps logical operators to logical operators. 
\end{proposition}
\begin{proof}
	The unitary operator $U$ may be decomposed in the Pauli bases: $U=\sum_{g,h} u_{gh}\sigma_{gh}$. Then, by proposition \ref{prop: Pauli Noise}, the effect of Pauli operators on logical operators can be determined, and by linearity the final result obtained.  
\end{proof}
Finally, we consider the case when a local noise channel acts on the $l^{th}$ physical qubit, thus resulting in a noisy resource state. This is modelled by a CPTP map with Kraus decomposition $\eta(\rho)=\sum_m K_m \rho K_m^\dag$, and its effect will be to map the logical superoperator associated with qubit $l$ to a new logical superoperator. The mapping in question given by the following proposition: 
\begin{proposition}\label{prop: CPTP}
If the $l^{th}$ qubit undergoes a CPTP map, represented by the quantum channel $\eta (\rho)=\sum_m K_m \rho K_m^\dag$, where each Kraus operator can be decomposed in the Pauli basis as $K_m=\sum_{a,b}^{(m)}k_{ab}\sigma_{ab}$, then the logical operator is mapped to $\tilde{\varepsilon}^{[i,j]}(\rho)=\sum_m \xi_m(A[i]) \rho \xi_m(A[i])$, where: 
\begin{equation}
\xi_m(A[i])=\sum_{a,b} k_{ab}^{(m)}\sigma_{0a}^{(m)} A[i] \sigma_{ab}^{(m)}. 
\end{equation}
\end{proposition}
\begin{proof}
The result is obtained by considering the Kraus representation of the CPTP map and decomposing it in the Pauli basis. By applying Proposition \ref{prop: Pauli Noise}, the final result is obtained. 
\end{proof}

Next, we consider first examples of Pauli channels which may act on the cluster state, and their effect within the correlation space. 
\begin{example} 
	The phase-flip channel has Kraus operators $K_s= \frac{1}{\sqrt{2}}Z^s$, $s=0,1$. The logical operator will thus evolve to $\xi_s(A[i])=\frac{1}{\sqrt{2}}A[i]Z^s$, and the logical superopertaor to: 
	\begin{equation}
	\tilde{\varepsilon}^{[i,j]}(\sigma)=\frac{1}{2}\sum_s A[i]Z^s \sigma  Z^sA[j].
	\end{equation}
\end{example}
\begin{example}
	The bit-flip channel has Kraus operators $K_s= \frac{1}{\sqrt{2}}X^s$, $s=0,1$. The logical operator will evolve to: $\xi_s(A[i])=\frac{1}{\sqrt{2}}Z^sA[i]X^s$, and thus the superoperator to: 
	\begin{equation}
	\tilde{\varepsilon}^{[i,j]}(\sigma)=\frac{1}{2}\sum_s Z^sA[i]X^s \sigma  X^sA[j] Z^s.
	\end{equation}
\end{example}

Next, we show how our framework allows for the study of error propagation by considering Pauli channels followed by a measurement in the $X$ basis. Note that it can be shown that if a qubit is measured in the $X$ basis, then its logical operator $A[i]$ is mapped to $\frac{1}{\sqrt{2}}HZ^m$, with $m$ denoting the measurement outcome.  
\begin{example} 
If a bit-flip channel is applied to the $k^{th}$ qubit, then its associated logical superoperator evolves to $\tilde{\varepsilon}^{[i_k,j_k]}(\rho )=\frac{1}{2}\sum_s Z^sA[i_k]X^s\rho X^sA[j_k]Z^s$. If the qubit is now measured in the $X$ basis, the associated logical superoperator to evolve to $\tilde{\varepsilon}^{[i_l,j_l]}(\rho )=\frac{1}{4}\sum_s Z^s HZ^m X^s\rho X^sZ^m H Z^s=\frac{1}{2} HZ^m \rho Z^m H$, and thus the output has not been affected by the channel. 
\end{example}

\begin{example}
The channel given by Kraus operators $K_0=\sqrt{p_0} \mathbb{I}$, $K_1=\sqrt{p_1}X$ and $K_2= \sqrt{p_2}Y$ is applied to the $k^{th}$ physical qubit. Thus, the associated superoperator evolves to: $\tilde{\varepsilon}^{[i_k, j_k]}=\sqrt{p_0}A[i_k]\rho A^\dag[j_k]+\sqrt{p_1}ZA[i_k]X\rho XA^\dag[j_k]Z+\sqrt{p_2}ZA[i_k]Y\rho Y A^\dag[j_k]Z$. Next, the qubit is measured in the $X$ basis, resulting in:  $\tilde{\varepsilon}^{[x_m, x_m]}=\frac{1}{2}(\sqrt{p_0}+ \sqrt{p_1}) HZ^m \rho Z^m H+\frac{1}{2}\sqrt{p_2}HZ^m Z \rho Z Z^m H$. This is equivalent to a $Z$ gate applied with probability $p_2$ to the qubit before the measurement. 
\end{example}

Finally, we consider the more general model of unitary noise by studying the effect of a Hadamard gate probabilistically applied to the physical qubit followed by an $X$ measurement, as in example \ref{ex: H meas}.
\begin{example}
The channel with Kraus operators $K_1=\sqrt{1-p} \mathbb{I}$ and $K_2=\sqrt{p}H$ is applied to the physical qubit. Thus, the associated logical superoperator evolves to $\tilde{\varepsilon}^{[i_k, j_k]}=\tilde{A}[i_k] \rho \tilde{A}^\dag [j_k]$, where:
\begin{equation}
\tilde{A}[i]=\frac{1}{\sqrt{2}}A[i]Z+\frac{1}{\sqrt{2}}ZA[i]X.
\end{equation} 
Next, the qubit is measured in the $X$ basis, and thus the logical operator is transformed to $\tilde{\varepsilon}^{[x_m, x_m]} (\rho)=\tilde{A}[x_m] \rho \tilde{A}^\dag [x_m]$, with: 
\begin{equation}
\tilde{A}[x_m]=\frac{1}{\sqrt{2}}HZ^m Z+\frac{1}{\sqrt{2}}ZHZ^mX, 
\end{equation} 
which reduces to $\tilde{A}[x_m]= \sqrt{2} H|k\rangle \langle k|$.
\end{example}

\section{Discussion} \label{sec:discussion}
%Recap 
Noise remains one of the major challenges in the quest to build a scalable quantum computer. In this paper, we have introduced two frameworks which allow for the effect of physical noise on the computation to be studied. More precisely, we consider the effect of single qubit local noise channels in the context of one-dimensional MBQC. This is first achieved in the circuit model, and next in the derived MPO representation. 

First, we considered teleportation in the presence of a noisy resource state, and found that the effect of the noise channel on the output could be determined. This then motivated us to consider the impact of noise in one-dimensional MBQC, as indeed, preparing pure states and performing clean measurements is experimentally challenging. Here, the introduced framework allowed us to understand the impact of having access to a mixed resource state on the computation, and the effect of performing noisy measurements. Moreover, this could be used in order to simulate noisy computation. 

%Future directions 
Thus, the next step would be to use these frameworks in order to represent noisy computations, and perform simulations. When MBQC is represented in the MPS formalism, the goal is to implement an operator $A$ via successive measurements in the correlation space. Due to the random nature of measurement outcomes, random Pauli by-product operators will occur throughout. By exploiting the properties of Clifford operators, we find that we are instead implementing the operator $UA$ where $U$ is a unitary operator depending on measurement outcomes and which may be subsequently corrected for. A similar analysis could be performed in the case of mixed states acted upon by logical superoperators. Another interesting question would be to study the numerical simulation of errors on the state and their impact on the computation. Overall, MPO offer us a flexible framework in which we can study how noise on cluster state computations affects the computation. Here, our simple examples illustrate its capabilities. Its full power might be in modelling the error channels that arise in experiments, and may prove to be useful for experimental groups in their modelling. 

An alternative direction for future research would be to consider more complex noise channels, and studying how the concatenation of building blocks will impact on the logical operators implemented by the computation. Indeed, the computation requires for precisely these blocks to be concatenated, sequentially implementing a logical operator on the teleported state. Now, in the presence of noise, errors will be layered throughout the computation. It would be useful to study how the measurement patterns, the logical operators and the noise channels relate to one another, and what is their effect on the computation output. Furthermore, a natural extension would be to study the generalisation of this framework to two-dimensions, by for instance considering grid-like structures \cite{gross2010quantum}. 

We have studied the effect of local noise channels on physical states and their impact on the computation performed. The former provides a better understanding of the role of errors in the computation, and the latter could allow for the simulation of noisy quantum computation. 
We hope that this work will stimulate further investigations of measurement-based quantum computations in the presence of noise and with mixed state resources.

\begin{acknowledgements}
We thank Matty Hoban and Hussain Anwar for stimulating discussions. This work was supported by the EPSRC.	
\end{acknowledgements}

\onecolumngrid  
\appendix 
\section{Proofs}
\subsection{Proof of Proposition $1$}
\begin{proof}
If we express the state on the first two qubits in the Bell basis, then the state of the joint system $\rho \otimes \Lambda $ is given by: 
\begin{equation}
\rho \otimes \Lambda =\frac{1}{4} \sum_{u,v} \sum_{i,j} \alpha_{uv} (-1)^{u k_1 +v k_2}|B_{u \oplus i, k_1} \rangle \langle B_{v \oplus j, k_2} |  \otimes \varepsilon (|i\rangle \langle j|).
\end{equation}
Next, the first two qubits are measured in the Bell basis $(P_{st} \otimes \mathbb{I})(\rho \otimes \Lambda)(P_{st} \otimes \mathbb{I})$, where $P_{st}=|P_{st} \rangle \langle P_{st}|$. The orthogonality of Bell states will demand that $s=u \oplus i= v \oplus j$ and $t=k_1=k_2$, and thus we have that: 
	\begin{align}
	(P_{st} \otimes \mathbb{I})(\rho \otimes \Lambda)(P_{st} \otimes \mathbb{I}) 	&=\frac{1}{4} \sum_{u,v,k,l} \alpha_{uv} (-1)^{(u+v)t}P_{st} \otimes K_{kl} X^s |u\rangle \langle v| X^s K_{kl}^\dag, \\
	&= \frac{1}{4} \sum_{u,v,k,l} \alpha_{uv} P_{st} \otimes K_{kl} X^s Z^t |u\rangle \langle v| Z^t X^s K_{kl}^\dag, \\
	&= \frac{1}{4} \sum_{k,l} P_{st} \otimes X^s Z^t  K_{kl} \rho K_{kl} ^\dag  Z^t X^s, \\
	&= \frac{1}{4} P_{st} \otimes X^s Z^t  \varepsilon (\rho )  Z^t X^s. 
	\end{align}
	We can see that the state $\varepsilon (\rho)$  has been teleported instead of the state $\rho $, and thus the quantum channel can be thought of as acting directly on the input state itself.  
\end{proof}
\subsection{Proof of Lemma $3$} \label{app: noisy cluster}
\begin{proof}

The quantum channel $\alpha_1(\rho)=\sum_m K_m \rho K_m^\dag$ is applied to the cluster state. Each Kraus operator $K_m$ can be decomposed in the Pauli basis as $K_m= a_{00}^{(m)}\mathbb{I} +a_{01}^{(m)} X+a_{10}^{(m)} Z -i a_{11}^{(m)} ZX$, where $a_{ij} \in \mathbb{C}$ for $i,j =0,1$. Here, we know that this channel is applied to the cluster and can thus exploit this information by considering the action of a Kraus operator $K_m$ on an $X$ eigenstate $|+\rangle$. It can easily be seen that this is equivalent to applying the operator $\tilde{K}_m=(a_{00}^{(m)} +a_{01}^{(m)})\mathbb{I} +( a_{10}^{(m)}-i a_{11}^{(m)})Z$ to the $|+\rangle$ state. By defining $\tilde{a}_{00}^{(m)}=a_{00}^{(m)} +a_{01}^{(m)}$ and $\tilde{a}_{10}^{(m)}= a_{10}^{(m)}-i a_{11}^{(m)}$, this new modified Kraus operator can be expressed as: $\tilde{K}_m=\tilde{a}_{00}^{(m)}\mathbb{I} + \tilde{a}_{10}^{(m)} Z$. 

Thus, the original quantum channel $\alpha_2(\rho)=\sum_m K_m \rho K_m^\dag$ has been mapped to a new channel $\tilde{\alpha}_2(\rho)=\sum_m \tilde{K}_m \rho \tilde{K}_m^\dag$, where $\tilde{K}_m=\sum_u \tilde{a}_{u0}^{(m)}Z^u$ and where $\tilde{a}_{u0}=\sum_v (-i)^{uv}a_{uv}$. 

Thus, if we consider the circuit shown in Fig. \ref{fig:alpha2}, the input is given by $|i_0 \rangle \langle i_0|\otimes |+\rangle \langle+|$. First, the second undergoes the noise channel $\alpha_2$ and which, given the previous argument, can now be expressed as: 
\begin{equation}
\big(|i_0 \rangle \langle j_0| \otimes  \alpha_2( |+\rangle \langle +|) \big)=\big(|i_0 \rangle \langle j_0| \otimes  \tilde{\alpha}_2( |+\rangle \langle +|) \big).
\end{equation}
Next, the two qubits are entangled via the control-$Z$ gate: 
\begin{equation}
\text{CZ}\big(|i_0 \rangle \langle j_0| \otimes  \alpha_2( |+\rangle \langle +|) \big)CZ = |i_0 \rangle \langle j_0| \otimes \sum_{m,u,v}  \tilde{a}_{u 0}^{(m)}  \tilde{a}_{v 0}^{(m)*} Z^uH|i_0\rangle \langle j_0|H Z^v,
\end{equation}
Finally, the first qubit is measured in the equatorial basis $|s_k\rangle =e^{-i\frac{\phi}{2}Z}Z^k |+\rangle$, with measurement outcome $k$ obtained:
\begin{equation}
 P_{sk}\text{CZ}\big(|i_0 \rangle \langle j_0| \otimes  \alpha_2( |+\rangle \langle +|) CZ\big) P_{sk}=\langle + | Z^k e^{i \frac{\phi}{2}Z}|i_0 \rangle \langle j_0| e^{-i\frac{\phi}{2}Z}Z^k |+\rangle P_{s_k} \otimes \sum_{m,u,v} \tilde{a}_{u 0}^{(m)}  \tilde{a}_{v 0}^{(m)*} Z^uH|i_0\rangle \langle j_0|H Z^v , 
\end{equation} 
where $P_{sk}= |s_k(\phi)\rangle \langle s_k(\phi)|$. This can be seen to be equal to the following expression: 
\begin{equation}
P_{sk}\text{CZ}\big(|i_0 \rangle \langle j_0| \otimes  \alpha_2( |+\rangle \langle +|) CZ\big) P_{sk}=	\frac{1}{2} e^{i \frac{\phi}{2}(-1)^{i_0}}e^{-i \frac{\phi}{2}(-1)^{j_0}}(-1)^{(i_0 + j_0)k}P_{sk} \otimes \sum_{m,u,v} \tilde{a}_{u 0}^{(m)}  \tilde{a}_{v 0}^{(m)*} Z^uH|i_0\rangle \langle j_0|H Z^v.
\end{equation}
By rearranging the coefficients onto the state of the second qubit, the original channel $\varepsilon_{k}$ can be recognised to have been applied to the original input $|i_0 \rangle \langle j_0|$:
\begin{equation}
P_{sk}\text{CZ}\big(|i_0 \rangle \langle j_0| \otimes  \alpha_2( |+\rangle \langle +|) CZ\big) P_{sk}=\frac{1}{2} P_{s_k} \otimes \sum_{m,u,v} \tilde{a}_{u 0}^{(m)}  Z^u \varepsilon_{k} (|i_0 \rangle \langle j_0|) ( \tilde{a}_{v 0}^{(m)} Z^v)^\dag, 
\end{equation}
which can be written as: 
\begin{equation}
P_{sk}\text{CZ}\big(|i_0 \rangle \langle j_0| \otimes  \alpha_2( |+\rangle \langle +|) CZ\big) P_{sk}=\frac{1}{2} P_{s_k} \otimes \sum_{m} \tilde{K}_m \varepsilon_{k} (|i_0 \rangle \langle j_0|) \tilde{K}_m^\dag. 
\end{equation}
Thus, the output on the second qubit is given by $\tilde{\alpha}_2 \circ \varepsilon_{k} (|i_0 \rangle \langle j_0|)$. 
\end{proof}	

\subsection{Proof of Lemma $4$} \label{app: rotated}
\begin{proof}
The quantum channel $\alpha_3(\rho)=\sum_m K_m \rho K_m^\dag$ is applied to the first qubit just before a measurement is performed. Each Kraus operator is decomposed in the rotated Pauli basis, where the rotation is around the $Z$ axis by an angle $\phi$, that is $e^{-i\frac{\phi}{2} Z} \sigma_{gh}e^{i \frac{\phi}{2}Z}$. Thus, we can express each Kraus operator as $K_m=\sum_m a_{gh}i^{gh} e^{-i\frac{\phi}{2} Z} X^g Z^h e^{i\frac{\phi}{2} Z}$ and thus $K_m^\dag=\sum_m a_{gh}^{(m)*}(-i)^{gh} e^{-i\frac{\phi}{2} Z}  Z^h X^g e^{i\frac{\phi}{2} Z}$. 

We now study the effect of the operator $K_m^\dag$ on an equatorial basis state $|s(\phi)_k \rangle =  e^{-i \frac{\phi}{2} Z} Z^k |+\rangle $:
\begin{align}
	K_m^\dag e^{-i \frac{\phi}{2} Z} Z^k |+\rangle &= \sum_{gh} a_{gh}^{(m)*} (-i)^{gh}  e^{-i\frac{\phi}{2} Z} Z^h X^g e^{i \frac{\phi}{2}Z} e^{-i \frac{\phi}{2} Z} Z^k |+\rangle,  \\
	&= \sum_{gh} a_{gh}^{(m)*} (-i)^{gh}(-1)^{gk} e^{-i \frac{\phi}{2}Z} Z^h Z^k |+\rangle, \\
	&=\sum_u \tilde{a}_{uk}^{(m)*} Z^u e^{-i \frac{\phi}{2}Z}Z^k |+\rangle,
\end{align}
where $\tilde{a}_{uk}^{(m)*}= \sum_{v} (-i)^{uv}(-1)^{kv}a_{vu}^{(m)*}$. This allows us to define the new modified Kraus operator $\tilde{K}_m^\dag= \sum_u \tilde{a}_{uk}^{(m)*} Z^u$, and associated quantum channel $\tilde{\alpha}_{3,k}$, where the index $k$ has been appended in order to emphasis the dependency on the measurement outcome $k$. 

Thus, the input to the circuit depicted in Fig. \ref{fig:noisy measurement} is $|i_0 \rangle \langle j_0| \otimes |+\rangle \langle +|$. The two qubits ae entangled by a control-$Z$ gate, and then the quantum channel $\alpha_3$ is applied to the first qubit, which can be transformed to $\tilde{\alpha}_3$:
\begin{equation}
	\sum_m \langle + | Z^k e^{i \frac{\phi}{2}Z}  \tilde{K}_m |i_0 \rangle \langle j_0| \tilde{K}_m^\dag e^{-i \frac{\phi}{2}Z}Z^k|+\rangle P_{s_k} \otimes H |i_0 \rangle \langle j_0| H,  
\end{equation} 
or equivalently: 
\begin{equation}
\sum_m a_{uk}^{(m)} a_{vk}^{(m)*} \langle + | Z^k e^{i \frac{\phi}{2}Z}  Z^u|i_0 \rangle \langle j_0|Z^v  e^{-i \frac{\phi}{2}Z}Z^k|+\rangle P_{s_k} \otimes H |i_0 \rangle \langle j_0| H. 
\end{equation} 
By expanding out this expression, we obtain: 
\begin{equation}
\frac{1}{2}	\sum_{m,u,v}   \tilde{a}_{uk}^{(m)} \tilde{a}_{vk}^{(m)} e^{i(-1)^{i_0}\frac{\phi}{2}} (-1)^{i_0(k+u)} e^{-i(-1)^{j_0} \frac{\phi}{2}}(-1)^{i_0(k+v)} P_{sk}\otimes H |i_0 \rangle \langle j_0| H.
\end{equation}
This can alternatively be expressed as: 
\begin{equation}
	\frac{1}{2} P_k \otimes \sum_{m,u,v} \tilde{a}_{uk}^{(m)} \tilde{a}_{vk}^{(m)*}  H Z^k e^{i \frac{\phi}{2} Z}  Z^u |i_0 \rangle \langle j_0 |Z^v e^{-i\frac{\phi}{2}Z}  Z^k H, 
\end{equation}
where we recognise the initial channel acting on the output: 
\begin{equation}
\frac{1}{2} P_k \otimes \sum_{m}  H Z^k  e^{i \frac{\phi}{2} Z} \tilde{K}_m |i_0 \rangle \langle j_0 | \tilde{K}_m^\dag e^{-i\frac{\phi}{2}Z} Z^k H. 
\end{equation}
Thus, this results in the channel $\varepsilon_k \circ \tilde{\alpha}_{3,k}$ being applied to the input state, where we note the dependency on measurement outcome for both channels. 
\end{proof}			
\subsection{Proof of theorem \ref{theorem: noise block}} \label{app: theorem}
\begin{proof}
In order to determine the output of Fig. \ref{fig:arbNoise}, we will first consider the simpler case depicted in Fig. \ref{fig:two noise}. 
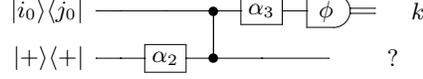
\begin{figure}[h!]
	\centering 
	\mbox{
		\Qcircuit @C=1em @R=.7em {
			\lstick{|i_0 \rangle \langle j_0|}	& \qw &\qw& \ctrl{1}  & \gate{\alpha_3} &\measureD{\phi} & \cw   & \rstick{k}   \\
			\lstick{|+\rangle \langle + |}   	 & \qw   &\gate{\alpha_2} & \ctrl{-1}   &\qw & \qw & \rstick{?}
		} }
		\caption[Noisy cluster state and measurement.]{\emph{Noisy cluster state and measurement.} The cluster state is affected by the noise channel $\alpha_2$ whilst the measurement will be affected by the noise channel $\alpha_3$.} 
		\label{fig:two noise}
	\end{figure}
The input to the circuit is given by $|i_0 \rangle \langle j_0| \otimes |+\rangle \langle +|$. Next, the noise channel $\alpha_2$ affects the cluster state, which can be transformed to $\tilde{\alpha}_2$ as discussed in lemma \ref{lemma: noisy resource}, before a control-$Z$ gate acts on the two qubits: 
\begin{equation}
\text{CZ}\big(|i_0 \rangle \langle j_0| \otimes  \alpha_2( |+\rangle \langle +|) \big)CZ = |i_0 \rangle \langle j_0| \otimes \sum_{m,u,v}  \tilde{a}_{u 0}^{(m)}  \tilde{a}_{v 0}^{(m)*} Z^uH|i_0\rangle \langle j_0|H Z^v.
\end{equation}
Next, the noise channel $\tilde{\alpha}_3$ is decomposed in the rotated Pauli basis, as discussed in lemma \ref{lemma: noisy meas} and applied to the first qubit: 	
\begin{equation}
\sum_{n,g,h} a_{gk}^{(n)} a_{hk}^{(n)*} \langle + | Z^k e^{i \frac{\phi}{2}Z}  Z^g|i_0 \rangle \langle j_0|Z^h  e^{-i \frac{\phi}{2}Z}Z^k|+\rangle P_{s_k} \otimes \sum_{m,u,v}  \tilde{a}_{u 0}^{(m)}  \tilde{a}_{v 0}^{(m)*} Z^uH|i_0\rangle \langle j_0|H Z^v.
\end{equation}
From here, we see that the coefficients can be directly moved onto the state of the second qubit, as in the proof of lemma \ref{lemma: noisy meas}: 
\begin{equation}
 P_{s_k} \otimes \sum_{m,u,v,n}  \tilde{a}_{u 0}^{(m)}  \tilde{a}_{v 0}^{(m)*} Z^ue^{i \frac{\phi}{2} Z} \tilde{K}_n |i_0 \rangle \langle j_0 | \tilde{K}_n^\dag e^{-i\frac{\phi}{2}Z} Z^k H Z^v.
\end{equation}
Thus, the channel $\tilde{\alpha}_2 \circ \varepsilon_k \circ \varepsilon_{3,k}$ was applied to the input. 
\end{proof}

In the more general case shown in Fig. \ref{fig:arbNoise}, the channel $\alpha_4 \circ \tilde{\alpha}_2 \circ \varepsilon_k \circ \varepsilon_{3,k} \circ \alpha_1 $ is applied to the input state. 

\subsection{Proof of Lemma $5$} \label{app: MPO}
\begin{proof}
	Consider a mixed state on $n$ qubits on a space $\mathcal{H}_{A}$, expressed in its diagonal basis: 
	\begin{equation}
	\rho^{[A]}=\sum_k p_k |v_k^{[A]}\rangle \langle v_k^{[A]}|.
	\end{equation}
	The Schmidt decomposition tells us that there exists an auxiliary system  $\mathcal{H}_{R}$, of identical dimension $2^n$, which we may add to the system. The resulting joint state $|\psi \rangle$ on the larger Hilbert space $\mathcal{H}_A \otimes \mathcal{H}_R$ is a pure state:
	\begin{equation}
	|\psi\rangle= \sum_k \sqrt{p_k} |v_k^{[A]} \rangle \otimes |v_k^{[R]}\rangle, 
	\end{equation}
	such that $\text{Tr}_R|\psi \rangle \langle \psi|=\rho{[A]}$. We know nothing about the entanglement induced by the purification between the system of interest and the reference system. We choose to picture them as shown in Fig. \ref{fig:Purify}, where we label the qubits of the system of interest from $1 \ldots n$, and the qubits from the reference system from $1' \ldots n'$.  
	\begin{figure}[h!] 
		\centering 
		\begin{tikzpicture}
		\definecolor{qu}{RGB}{77, 77, 255}
		\definecolor{cz}{RGB}{254, 68, 0}
		
		%\node[fill, scale = 0.7] at (0,1) [circle, draw]{};		
		\node[fill, scale = 0.7] at (6,0) [circle,draw] {};
		\node[fill, scale = 0.7] at (4.5,0) [circle,draw] {};
		\node[fill, scale = 0.7] at (3,0) [circle,draw] {};
		\node[fill, scale = 0.7] at (1.5,0) [circle,draw] {};
		\node[fill, scale = 0.7] at (0,0) [circle,draw] {};
		\node[qu, fill=qu, scale = 0.7] at (0,2) [circle,draw] {};
		\node[qu, fill=qu, scale = 0.7] at (1.5,2) [circle,draw] {};
		\node[qu, fill=qu, scale = 0.7] at (3,2) [circle,draw] {};
		\node[qu, fill=qu, scale = 0.7] at (4.5,2) [circle,draw] {};
		\node[qu, fill=qu, scale = 0.7] at (6,2) [circle,draw] {};
		\draw (0,1) node[dotted, minimum height=3cm,minimum width=1cm,draw] {};
		\draw (1.5,1) node[dotted, minimum height=3cm,minimum width=1cm,draw] {};
		\draw (3,1) node[dotted, minimum height=3cm,minimum width=1cm,draw] {};
		\draw (4.5,1) node[dotted, minimum height=3cm,minimum width=1cm,draw] {};
		\draw (6,1) node[dotted, minimum height=3cm,minimum width=1cm,draw] {};
		\end{tikzpicture}
		\caption[Purification.]{\emph{Purification.}Two systems, in black the system of interest and in blue the reference system.} 
		\label{fig:Purify}
	\end{figure}
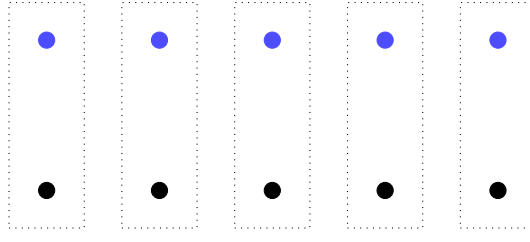
	
	The idea is now to consider the joint system of a real and auxiliary qubit, and apply the same proof as Vidal to these objects. Thus, we first partition the qubits between the utmost left real and auxiliary qubits, and the other $2(n-1)$ qubits. A Schmidt decomposition is then performed across this partition: 
	\begin{equation}
	|\psi \rangle = \sum_{\alpha_1} \lambda_{\alpha_1}^{[1,1']} |\Phi_{\alpha_1}^{[1,1']}\rangle \otimes |\Phi_{\alpha_1}^{[2,2' \ldots, n, n']}\rangle,
	\end{equation}
	where there are $\chi_{[1,1']}$ terms in this sum, a quantity proportional to the amount of entanglement present between the first qubits of the initial and reference systems and the rest. The first Schmidt vector is expanded in the computational basis, where the index $i$ is used for the real qubit and $s$ for the virtual qubit:
	\begin{equation}
	|\Phi_{\alpha_1}^{[1,1']}\rangle= \sum_{i_1, s_1} \Gamma_{\alpha_1}^{[1,1'] i_1 s_1} |i_1 s_1\rangle,
	\end{equation}
	thus yielding the state:
	\begin{equation}
	|\psi\rangle= \sum_{\alpha_1, i_1} \lambda_{\alpha_1}^{[1,1']}  \Gamma_{\alpha_1}^{[1,1'] i_1 s_1} |i_1s_1\rangle \otimes  |\Phi_{\alpha_1}^{[2 \ldots n']} \rangle .
	\end{equation}
	Now, the second Schmidt vector is expressed as: 
	\begin{equation}
	|\Phi_{\alpha_1}^{[2 \ldots n']}\rangle = \sum_{i_2 s_2} |i_2 s_2\rangle \otimes |\tau_{\alpha_1 i_2 s_2}^{[3 \ldots n']}\rangle,
	\end{equation}
	and thus $|\psi\rangle$ may be written: 
	\begin{equation}
	|\psi\rangle= \sum_{\alpha_1, i_1, i_2} \lambda_{\alpha_1}^{[1,1']}  \Gamma_{\alpha_1}^{[1,1'] i_1 s_1} |i_1 s_1\rangle \otimes |i_2 s_2\rangle \otimes |\tau_{\alpha_1 i_2 s_2}^{[3 \ldots n']}\rangle.
	\end{equation}
	The arbitrary state on the qubits $[3, \ldots, n']$ is expressed in the Schmidt bases for the partition $[3, \ldots, n']$:
	\begin{equation}
	|\tau_{\alpha_1 i_2 s_2}^{[3 \ldots n']}\rangle=\sum_{\alpha_2} \Gamma_{\alpha_1 \alpha_2}^{[2,2']i_2 s_2} \lambda_{\alpha_2}^{[2,2']}|\Phi_{\alpha_2}^{[3 \ldots n']}\rangle, 
	\end{equation}
	which thus give: 
	\begin{equation}
	|\psi\rangle= \sum_{\alpha_1, i_1, i_2, \alpha_2} \lambda_{\alpha_1}^{[1,1']}  \Gamma_{\alpha_1}^{[1,1'] i_1}  \Gamma_{\alpha_1 \alpha_2}^{[2,2']i_2 s_2} \lambda_{\alpha_2}^{[2,2']} |i_1 s_1\rangle \otimes |i_2 s_2\rangle \otimes |\Phi_{\alpha_2}^{[3 \ldots n']}\rangle.
	\end{equation}
	This process is repeated, until finally we obtain: 
	\begin{equation}
	|\psi\rangle= \sum_{\alpha_1, i_1\ldots  \alpha_n} \lambda_{\alpha_1}^{[1,1']}  \Gamma_{\alpha_1}^{[1,1'] i_1 s_1}  \Gamma_{\alpha_1 \alpha_2}^{[2,2']i_2 s_2} \lambda_{\alpha_2}^{[2,2']} \ldots \Gamma_{\alpha_n}^{[n,n']i_n s_n}|i_1 s_1 \ldots i_n s_n\rangle.
	\end{equation}
	The Schmidt coefficients are then absorbed into the tensors: 
	\begin{equation}
	|\psi\rangle= \sum_{\alpha_1, i_1\ldots  \alpha_n} \tilde{ \Gamma}_{\alpha_1}^{[1,1'] i_1 s_1} \tilde{ \Gamma}_{\alpha_1 \alpha_2}^{[2,2']i_2 s_2}  \ldots\tilde{ \Gamma}_{\alpha_n}^{[n,n']i_n s_n}|i_1 s_1 \ldots i_n s_n\rangle.
	\end{equation}
	The alpha indices can be thought of as implementing the multiplication between these different tensors, and thus, by associating vectors $v$ with tensors with one index and matrices $A$ for those with two, we obtain:
	\begin{equation}
	|\psi\rangle= \sum_{ i_1 s_1 \ldots i_n s_n}   v [i_1, s_1] A[i_2, s_2] \dots v[i_n, s_n] |i_1 s_1 \ldots i_n s_n\rangle, 
	\end{equation}
	where we have also assumed translation independence. By relabelling, we finally get: 
	\begin{equation}
	|\psi\rangle= \sum_{ i_1 s_1 \ldots i_n s_n}   v[i_n, s_n] A[i_{n-1}, s_{n-1}] \dots v[i_1, s_1] |i_1 s_1 \ldots i_n s_n\rangle.
	\end{equation}
	These tensors now are not only dependent on the state of the real qubits, but also on the state of the auxiliary qubits. Whereas previously the dimensions of the matrices were upper bounded by the maximum Schmidt rank over the $(n-1)$ partitions, their dimension is now dependent on the amount of entanglement induced by the reference system.
	
	Next, we can trace over the reference system and obtain an expression for the state of the system, and by defining the matrix $\rho[i_1, j_1]=\sum_{s_1} v^\dag [i_1, s_1] v[j_1, s_1]$, where $v[i_1, s_1] $ denotes a line vector and $\vec{v}^\dag[j_1, s_1]$ a column vector, we thus obtain: 
	\begin{equation}
	\rho^{[A]}= \sum_{s} \sum_{i} v[i_n, s_n] \dots  A[i_2,s_2]\rho [i_1, j_1]A^\dag[j_2,s_2]  \dots  v^\dag[j_n, s_n] |i_1 \ldots i_n \rangle \langle j_1 \ldots j_n|.
	\end{equation}
	The superoperator $ \varepsilon^{[i_{k}, j_{k}]}$ is defined as: 
	\begin{equation}
	\varepsilon^{[i_{k}, j_{k}]}(\rho)= \sum_s A[i_{k},s] \rho A^\dag[j_{k},s], 
	\end{equation}
	and 
	\begin{equation}
	\sigma^{[i,j]}(A)=\sum_s v^\dag[i, s] A v[j, s].
	\end{equation}
	Thus, the mixed state is finally given by: 
	\begin{equation}
	\rho^{A}= \sum \sigma^{[i_n, j_n]} \circ \varepsilon^{[i_{n-1}, j_{n-1}]}\circ \dots  \varepsilon^{[i_{n-1}, j_{n-1}]}(\rho [i_1, j_1])  |i_1 \ldots i_n \rangle \langle j_1 \ldots j_n|.
	\end{equation}
	
	This initial state is thus acted upon by sequence of $n-2$ superoperators denoted by $\varepsilon$, and where $\circ$ denotes the composition of superoperators. Finally, the operator $\sigma^{i_n, j_n}$ acts on the evolved logical state, and is analogous to a measurement. The key difference with MPS is the additional index $s$ which arise in the logical operators $A[k,s]$. 
\end{proof}

\subsection{Proof of Proposition $3$} \label{app: pauli on MPO}
\begin{proof}
	The MPS representation of a cluster state is given by $|\psi\rangle=\sum \langle i_n| A[i_{n-1}] \ldots A[i_2] A[i_1]|+\rangle |i_1 \ldots i_n\rangle$. If a Pauli $X$ is applied to the $j^{th}$ qubit, then the state evolves to:
	\begin{equation}
	X_j |\psi \rangle =\sum(-1)^{i_n i_{n-1}} \ldots (-1)^{ i_{j+1}i_j} (-1)^{i_ji_{j-1}}\ldots (-1)^{i_1 i_2}|i_1 \ldots i_{j-1}\rangle |i_j \oplus 1\rangle |i_{j+1}\ldots i_n\rangle. 
	\end{equation}
	By relabelling: 
	\begin{align}
	X_j |\psi\rangle&=\sum(-1)^{i_n i_{n-1}} \ldots  (-1)^{ i_{j+1}(i_j \oplus1)} (-1)^{(i_j \oplus 1)i_{j-1}} \ldots (-1)^{i_1 i_2}|i_1 \ldots i_{j-1}\rangle |i_j\rangle |i_{j+1}\ldots i_n\rangle,\\
	&= \sum \langle i_n| A[i_{n-1}] \ldots H|i_{j+1}\rangle \langle i_{j+1}| H |i_j \oplus 1\rangle \langle i_j \oplus 1\rangle H|i_{j-1}\rangle \langle i_{j-1}| \ldots A[i_1]|+\rangle |i_1 \ldots i_n\rangle ,\\
	&= \sum \langle i_n| A[i_{n-1}] \ldots A[i_{j+1}] ZA[i_j]X A[i_{j-1}] \ldots A[i_1]|+\rangle |i_1 \ldots i_n\rangle, 
	\end{align}
	and we can interpret this as the logical operator acting on the $j$th qubit being mapped to $ZA[k]X=\sigma_{01}A[i_j]\sigma_{10}$. 
	
	If a Pauli $Z$ is applied to the $j^{th}$ qubit, then the state evolves to: 
	\begin{equation}
	Z_j |\psi \rangle=\sum(-1)^{i_n i_{n-1}} \ldots (-1)^{ i_{j+1}i_j} (-1)^{i_ji_{j-1}}\ldots (-1)^{i_1 i_2} (-1)^{i_j}|i_1 \ldots i_n\rangle.
	\end{equation}
	The effect on the logical operators can be expressed in two ways. Either: 
	\begin{align}
	Z_j |\psi\rangle&=\sum \langle i_n| \ldots A[i_{j+1}] H|i_j\rangle \langle i_j|Z A[i_{j-1}] \ldots |+\rangle  |i_1 \ldots i_n\rangle,\\
	&=\sum \langle i_n|\ldots A[i_{j+1}] A[i_j] Z A[i_{j-1}] \ldots|+\rangle  |i_1 \ldots i_n\rangle,
	\end{align}
	where we can now define $\tilde{A}[k]=A[k]Z$, or alternatively: 
	\begin{align}
	Z_j |\psi\rangle&=\sum  \langle i_n|A[i_{n-1}] \ldots A[i_{j+1}] XH|i_j\rangle \langle i_j| A[i_{j-1}] \ldots A[i_1]|+\rangle   |i_1 \ldots i_n\rangle, \\
	&  =\sum \langle i_n|A[i_{n-1}] \ldots A[i_{j+1}] XA[i_j] A[i_{j-1}] \ldots A[i_1]|+\rangle   |i_1 \ldots i_n\rangle,
	\end{align}
	where we can now define $\tilde{A}[k]=XA[k]$. Thus, when the Pauli $Z=\sigma_{01}$ operator is applied, the logical operators can be represented by two equivalent evolutions: $A[i_j]Z=A[i_j]\sigma_{01}=XA[i_j]=\sigma_{10}A[i_j]$. Finally, if a Pauli operator $Y=iXZ$ is applied to the $j^{th}$ qubit, then the state evolves to: 
	\begin{equation}
	Y_j |\psi \rangle = \sum (-1)^{i_n i_{n-1}} \ldots (-1)^{i_2 i_1}  i(-1)^{i_j} |i_1 \ldots i_{j-1}\rangle |i_j \oplus 1\rangle |i_{j+1}\ldots i_n\rangle.
	\end{equation} 
	Once again, its effect on the logical operators can be expressed in two ways: 
	\begin{align}
	Y_j |\psi\rangle&=\sum (-1)^{i_n i_{n-1}} \ldots (-1)^{ i_{j+1}i_j} (-1)^{i_ji_{j-1}}\ldots (-1)^{i_1 i_2}  (-1)^{i_j}|i_1 \ldots i_{j-1}\rangle|i_j \oplus 1\rangle |i_{j+1}\ldots i_n\rangle,\\
	&= \sum \langle i_n | A[i_{n-1}] \ldots H|i_j\rangle \langle i_j|Z \ldots A[i_1]|+\rangle |i_1, \ldots i_{j-1}\rangle |i_j \oplus 1\rangle |i_{j+1} \ldots i_n\rangle, 
	\end{align}
	which if we relabel: 
	\begin{align}
	Y_j |\psi\rangle&= \sum \langle i_n |  \ldots H|i_j \oplus 1\rangle \langle i_j \oplus 1|Z \ldots |+\rangle |i_1, \ldots i_{j-1}\rangle |i_j\rangle |i_{j+1} \ldots i_n\rangle,\\
	&=\sum \langle i_n |  \ldots ZA[i_j]XZ \ldots |+\rangle |i_1 \ldots  i_n\rangle,
	\end{align}
	where we can now define $\tilde{A}[i_j]=ZA[i_j]XZ=ZA[i_j]Y$.
	Alternatively, 
	\begin{align}
	Y_j |\psi\rangle &= \sum \langle i_n | \ldots HZ|i_j \oplus 1\rangle \langle i_j \oplus 1| \ldots|+\rangle |i_1, \ldots i_{j-1}\rangle |i_j\rangle |i_{j+1} \ldots i_n\rangle,\\
	&= \sum \langle i_n | \ldots HZX|i_j\rangle \langle i_j| X \ldots |+\rangle |i_1\ldots i_n\rangle,\\
	&= \sum \langle i_n | \ldots XZA[i_j]X \ldots |+\rangle |i_1 \ldots i_n\rangle,
	\end{align}
	where we can now define $\tilde{A}[i_j]=iXZA[i_j]X=YA[i_j]X$.
	Thus, when the Pauli operator $Y=\sigma_{11}$ acts on the state, the logical operators evolve to: $ZA[i_j]XZ=ZA[i_j]Y=\sigma_{01}A[i_j]\sigma_{11}=YA[i_j]X=\sigma_{11}A[i_j]\sigma_{10}$.
\end{proof}
\subsection{Proof of Proposition $5$}
\begin{proof}
The Kraus operators $K_m$ are decomposed in the Pauli bases: $K_m=\sum_{a,b} k_{ab}^{(m)}\sigma_{ab}^{(m)}$. Thus, when the CPTP map acts on the $l^{th}$ qubit, we have:
\begin{equation}
\eta_j(\rho)=\sum_m \sum_{i, j} \sum_{a,b,g,h}\sigma^{[i_n,j_n]}\circ  \ldots \circ \varepsilon^{[i_2, j_2]} ( |+ \rangle \langle +| )   |i_1\rangle \langle j_1| \ldots k_{ab}^{(m)} k_{gh}^{(m)*} \sigma_{ab}^{(m)} (|i_l\rangle \langle j_l|)\sigma_{gh}^{(m)\dag} \ldots |i_n\rangle \langle j_n|.
\end{equation}
By linearity, and applying proposition \ref{prop: Pauli Noise}, this can be expressed as: 
\begin{equation}
\eta_j(\rho)= \sum_{i, j} \sigma^{[i_n,j_n]}\circ \varepsilon^{[i_{n-1} j_{n-1}]} \circ \ldots \circ \tilde{\varepsilon}^{[i_l, j_l]} \circ \ldots \circ \varepsilon^{[i_2 j_2]} ( |+ \rangle \langle +| )  |i_1 \ldots i_n\rangle \langle j_1 \ldots j_n|,
\end{equation}
where we have defined $\tilde{\varepsilon}^{[i,j]}(\rho)=\sum_{m}  \xi_m(A[i])  \rho \xi_m^\dag(A[j])$, and where $\xi_m(\rho)=\sum_{a,b} k_{ab}^m \sigma_{0a}^{(m)} \rho  \sigma_{ab}^{(m)}$.
\end{proof}
\end{document}